\documentclass{article}
\usepackage{authblk}
\usepackage[T1]{fontenc}
\usepackage{soul} 
\usepackage{phfqit}
\usepackage{proba} 
\usepackage{amssymb}
\usepackage{amsmath}
\usepackage{amsthm}
\usepackage{amsfonts}
\usepackage{mathtools}
\usepackage{xspace}
\usepackage{bm}
\usepackage{nicefrac}
\usepackage{commath}
\usepackage{bbm}
\usepackage{boxedminipage}
\usepackage{xparse}
\usepackage{xcolor}
\usepackage{float}
\usepackage{multirow}
\usepackage{graphicx}
\usepackage{caption}
\usepackage{subcaption}
\usepackage{ifthen}
\usepackage{adjustbox}
\usepackage[ruled,vlined,linesnumbered]{algorithm2e}
\usepackage{colortbl} 
\usepackage{fancyhdr}   
\usepackage{hyperref}
    \hypersetup{ colorlinks=true, linkcolor=blue, filecolor=magenta, urlcolor=red, citecolor=blue, breaklinks}
\usepackage{fullpage}
\usepackage[utf8]{inputenc}
\usepackage{environ}
\usepackage[colorinlistoftodos]{todonotes}
\usepackage{mdframed}
\usepackage{setspace}
\usepackage{microtype}

\usepackage[capitalize]{cleveref}

\usepackage{tikz}
\usetikzlibrary{shapes}
\usetikzlibrary{positioning}
\usetikzlibrary{fit}
\usetikzlibrary{calc}
\usetikzlibrary{backgrounds}
\usetikzlibrary{patterns}
\usetikzlibrary{matrix}
\usetikzlibrary{arrows}

\usepackage[noadjust]{cite}

\newtheorem{proposition}{Proposition}
\newtheorem{lemma}{Lemma}
\newtheorem{theorem}{Theorem}
\newtheorem{corollary}{Corollary}
\newtheorem{definition}{Definition}

\newtheorem{claim}{Claim}
\theoremstyle{remark}
\newtheorem{remark}{Remark}

\newcommand{\ie}{\text{i.e.}\xspace}
\newcommand{\etal}{\text{et al.}\xspace}

\newcommand{\polylog}{\ensuremath{\mathrm{polylog}}\xspace}
\newcommand{\negl}{\ensuremath{\mathsf{negl}}\xspace}

\newcommand{\zo}{\ensuremath{{\{0,1\}}}\xspace}

\newcommand{\tuple}[1]{\ensuremath{\left(#1\right)}}
\newcommand{\tp}{\tuple}
\newcommand{\eps}{\ensuremath{\varepsilon}}

\newcommand{\ceil}[1]{\ensuremath{\left\lceil{#1}\right\rceil}\xspace}
\newcommand{\floor}[1]{\ensuremath{\left\lfloor{#1}\right\rfloor}\xspace}

\newcommand{\itn}[1]{^{(#1)}}
\def\*#1{\mathbf{#1}}
\def\+#1{\mathcal{#1}}

\NewEnviron{problem}[1]{%
	\begin{center}\fbox{\parbox{6in}{%
				{\centering\scshape #1\par}%
				\parskip=1ex
				\everypar{\hangindent=1em}%
				\BODY
}}\end{center}}

    \newcommand{\Dec}{\ensuremath{\mathsf{Dec}}\xspace}

    \newcommand{\ED}{\ensuremath{\mathsf{ED}}\xspace}

    \newcommand{\ignore}[1]{{}}
    
    \newcommand{\NAT}{\ensuremath{\mathbb{N}}}

    \newcommand{\dist}{\mathrm{dist}\hspace{-1pt}}

\newcommand{\HAM}{\ensuremath{\mathsf{HAM}}}

\def\final{1} 
\ifnum\final=0  
\newcommand{\authnote}[3]{\textcolor{#2}{{\sf (#1's Note: {\sl{#3}})}}}
\newcommand{\jeremiah}{\authnote{Jeremiah}{blue}}
\newcommand{\minshen}{\authnote{Minshen}{purple}}
\newcommand{\xnote}{\authnote{Xin}{magenta}}
\newcommand{\enote}{\authnote{Elena}{green}}
\newcommand{\knote}{\authnote{Kuan}{cyan}}
\newcommand{\yu}{\authnote{Yu}{cyan}}
\newcommand{\alex}[1]{\authnote{Alex}{orange}{#1}}
\else 
\newcommand{\authnote}{}
\newcommand{\jeremiah}[1]{}
\newcommand{\minshen}[1]{}
\newcommand{\xnote}[1]{}
\newcommand{\enote}[1]{}
\newcommand{\knote}[1]{}
\newcommand{\yu}[1]{}
\newcommand{\alex}[1]{}
\fi

\newif\ifnotes\notestrue 

\newcommand{\tildex}{\ensuremath{\widetilde{x}}\xspace}
\newcommand{\tildey}{\ensuremath{\widetilde{y}}\xspace}
\newcommand{\tildec}{\ensuremath{\widetilde{c}}\xspace}

\newcommand{\Cout}{C_{\text{out}}}
\newcommand{\dout}{\delta_{\text{out}}}
\newcommand{\qout}{q_{\text{out}}}
\newcommand{\epsout}{\eps_{\text{out}}}
\newcommand{\gammaout}{\gamma_{\text{out}}}
\newcommand{\Cin}{C_{\text{in}}}
\newcommand{\din}{\delta_{\text{in}}}
\newcommand{\rin}{r_{\text{in}}}
\newcommand{\ham}{\HAM\xspace}

\newcommand{\lcs}{\ensuremath{\mathrm{LCS}}\xspace}
\newcommand{\LCS}{\ensuremath{\mathbf{LCS}}\xspace}
\newcommand{\Span}{\textsf{Span}}
\usepackage[utf8]{inputenc}
\usepackage{thm-restate}

\title{On Relaxed Locally Decodable Codes for Hamming and  Insertion-Deletion Errors}
\author[1]{Alex Block\thanks{Supported by NSF Award CCF-1910659}}
\author[1]{Jeremiah Blocki\thanks{Supported by NSF CAREER Award CNS-2047272 and NSF Award CCF-1910659,}}
\author[2,3]{Kuan Cheng}
\author[1]{Elena Grigorescu\thanks{Supported by NSF CCF-1910659, NSF CCF-1910411, and NSF CCF-2228814.}}
\author[4]{Xin Li \thanks{Supported by NSF CAREER Award CCF-1845349 and NSF Award CCF-2127575.}}
\author[4]{Yu Zheng \thanks{Supported by NSF CAREER Award CCF-1845349.}}
\author[1]{Minshen Zhu \thanks{Supported by NSF CCF-1910659, NSF CCF-1910411, and NSF CCF-2228814.}}
\affil[1]{Department of Computer Science, Purdue University}
\affil[2]{Center on Frontiers of Computing Studies, Peking University}
\affil[3]{Advanced Institute of Information Technology, Peking University}
\affil[4]{Department of Computer Science, Johns Hopkins University}
\affil[ ]{\textit {\{jblocki, elena-g, zhu628\}@purdue.edu}}
\affil[ ]{\textit {ckkcdh@pku.edu.cn} }
\affil[ ]{\textit {\{lixints, yuzheng\}@cs.jhu.edu} }

\date{}

\begin{document}

\maketitle

\begin{abstract}
 
  Locally Decodable Codes (LDCs) are error-correcting codes $C:\Sigma^n\rightarrow \Sigma^m,$ encoding {\em messages} in $\Sigma^n$ to {\em codewords} in $\Sigma^m$,  with super-fast decoding algorithms. They are important mathematical objects in many areas of theoretical computer science, yet the best constructions so far have codeword length $m$ that is super-polynomial in $n$, for codes with constant query complexity and constant alphabet size.

  In a very surprising result, Ben-Sasson, Goldreich, Harsha, Sudan, and Vadhan (SICOMP 2006) show how to construct a relaxed version of LDCs (RLDCs) with constant query complexity and almost linear codeword length over the binary alphabet, and used them to obtain significantly-improved constructions of Probabilistically Checkable Proofs. 
  
  In this work, we study RLDCs in the standard Hamming-error setting, and introduce their variants in the insertion and deletion (Insdel) error setting. Standard LDCs for Insdel  errors were first studied by Ostrovsky and Paskin-Cherniavsky ({\em Information Theoretic Security, 2015}), and are further motivated by recent advances in DNA random access bio-technologies (Banal et al., {\em Nature Materials, 2021}), in which the goal is to retrieve individual files from a DNA storage database.
  
   Our first result is an exponential lower bound on the length of Hamming RLDCs making $2$ queries (even adaptively), over the binary alphabet. This answers a question explicitly raised by Gur and Lachish (SICOMP 2021) and is the first exponential lower bound for RLDCs. Combined with the results of Ben-Sasson et al., our result exhibits a ``phase-transition''-type behavior on the codeword length for some constant-query complexity. We achieve these lower bounds via a transformation of RLDCs to standard Hamming LDCs, using a careful analysis of restrictions of message bits that fix codeword bits.

   We further define two variants of RLDCs in the Insdel-error setting, a weak and a strong version. On the one hand, we  construct weak Insdel RLDCs with almost linear codeword length and constant query complexity, matching the parameters of the Hamming variants. On the other hand, we prove exponential lower bounds for strong Insdel RLDCs. These results demonstrate that, while these variants are equivalent in the Hamming setting, they are significantly different in the insdel setting. Our results also prove a strict separation between Hamming RLDCs and Insdel RLDCs.
   

\end{abstract}

\section{Introduction}\label{sec:intro}

Locally Decodable Codes (LDCs) \cite{KatzT00, SudanTV99} are error-correcting codes $C: \Sigma^n \rightarrow \Sigma^m$ that have super-fast decoding algorithms that can recover individual symbols of a {\em message} $x\in \Sigma^n$, even when worst-case errors are introduced in the {\em codeword} $C(x)$. Similarly, Locally Correctable Codes (LCCs) are error-correcting codes $C: \Sigma^n \rightarrow \Sigma^m$ for which there exist very fast decoding algorithms that recover individual symbols of the {\em codeword} $C(x)\in \Sigma^m$, even when worst-case errors are introduced. LDCs/LCCs were first discovered by Katz and Trevisan \cite{KatzT00} and since then 
have proven to be crucial tools in many areas of computer science, including private information retrieval, probabilistically checkable proofs, self-correction,  fault-tolerant circuits,  hardness amplification, and data structures (e.g., \cite{BabaiFLS91,LundFKN92,BlumLR93,BlumK95,ChorKGS98,ChenGW13,AndoniLRW17} and surveys \cite{Tre04-survey,Gasarch04}).

The {\em parameters} of interest of these codes are their {\em  rate}, defined as the ratio between the message length $n$ and the codeword length $m$,  their {\em relative minimum distance}, defined as the minimum normalized Hamming distance between any pair of codewords, and their {\em locality} or {\em query complexity}, defined as the number of queries a decoder makes to a received word  $y\in \Sigma^m$. Trade-offs between the achievable parameters of Hamming LDCs/LCCs have been studied extensively over the last two decades \cite{KerenidisW04, WehnerW05, GoldreichKST06, Woodruff07, Yekhanin08, Yekhanin12, DvirGY11, Efremenko12,GalM12, BhattacharyyaDS16, BhattacharyyaG17, bhattacharyya2017lower, DvirSW17,KoppartyMRS17, BhattacharyyaCG20} (see also surveys by Yekhanin \cite{Yekhanin12} and by Kopparty and Saraf \cite{KoppartyS16}). 

Specifically, for $2$-query Hamming LDCs/LCCs it is known that $m=2^{\Theta(n)}$ \cite{KerenidisW04, GoldreichKST06, Ben-AroyaRW08, bhattacharyya2017lower}. However, for $q>2$ queries, the current gap between upper and lower bounds is superpolynomial in $n$. In particular, the best constructions have super-polynomial codeword length \cite{Yekhanin08,DvirGY11,Efremenko12}, 
while the most general lower bounds for $q\geq 3$ are of the form $m=\Omega((\frac{n}{\log n})^{1+1/(\ceil{\frac{q}{2}}-1)})$ \cite{KatzT00,KerenidisW04}. In particular, for $q=3$, \cite{KatzT00} showed an $m=\Omega(n^{3/2})$ bound, which was improved in \cite{KerenidisW04} to $m=\Omega(n^2/\log^2 n)$. This was further improved by \cite{Woodruff07,Woodruff12} to $m=\Omega(n^2/\log n)$ for general codes and $m=\Omega(n^2)$ for linear codes. \cite{bhattacharyya2017lower} used new combinatorial techniques to obtain the same $m=\Omega(n^2/\log n)$ bound. A very recent paper \cite{AlrabiahGKM} breaks the quadratic barrier and proves that  $m=\Omega(n^3/\poly\log n)$.
 We note that the  exponential lower bound on the length of $3$-query LDCs  from \cite{GalM12}  holds only for some restricted parameter regimes, and do not apply to the natural ranges of the known upper bounds.

Motivated by this large gap in the constant-query regime, as well as by applications in constructions of Probabilistically Checkable Proofs (PCPs), Ben-Sasson, Goldreich, Harsha, Sudan, and Vadhan \cite{Ben-SassonGHSV06} introduced a relaxed version of LDCs for Hamming errors. Specifically, the decoder is allowed to output a ``decoding failure'' answer (marked as ``$\bot$''), as long as it errs with some small probability. More precisely,  a \emph{$(q,\delta, \alpha, \rho)$-relaxed LDC} is an error-correcting code satisfying the following properties.

\begin{definition}\label{def:strongRLDC} A $(q,\delta, \alpha, \rho)$-Relaxed Locally Decodable Code 
${C}: \Sigma^n \rightarrow \Sigma^m$ is a code for which there exists a decoder that makes at most $q$ queries to the received word $y$, and satisfies the following further properties:
\begin{enumerate} 
\item  (Perfect completeness) For every $i\in [n]$, if $y=C(x)$ for some message $x$ then the decoder, on input $i$, outputs $x_i$ with probability $1.$
\item  (Relaxed decoding) For every $i\in [n]$, if $y$ is such that $dist(y, C(x))\leq \delta $ for some unique $C(x)$, then the decoder, on input $i$, outputs $x_i$ or $\bot$ with probability   $\geq \alpha$. 
\item (Success rate) For every $y$ such that $dist(y, C(x))\leq \delta $ for some unique $C(x)$, there is a set $I$ of size $\geq \rho n$ such that for every $i\in I$ the decoder, on input $i$, correctly outputs $x_i$ with probability $\geq \alpha$. 
\end{enumerate}
We will call an RLDC that satisfies all $3$ conditions by the notion of {\em strong} RLDC, and one that satisfies just the first $2$ conditions by the notion of {\em weak} RLDC, in which case it is called a $(q,\delta, \alpha)$-RLDC. Furthermore,  if the $q$ queries are made in advance, before seeing entries of the codeword, then the decoder is said to be {\em non-adaptive}; otherwise, it is called {\em adaptive}.
\end{definition}




The above definition is quite general, in the sense that $dist(a,b)$ can refer to several different distance metrics. In the most natural setting, we use $dist(a,b)$ to mean the ``relative'' Hamming distance between $a,b\in \Sigma^m$, namely $dist(a,b)=|\{i\colon a_i\ne b_i\}|/m$. This corresponds to the standard RLDCs for Hamming errors. As it will be clear from the context, we also use $dist(a,b)$ to mean the ``relative'' Edit distance between $a,b \in \Sigma^*$,  namely $dist(a, b)=\ED(a, b)/(|a|+|b|)$, where $\ED(a, b)$ is the minimum number of insertions and deletions to transform string $a$ into $b$. This corresponds to the new notion introduced and studied here, which we call {\em  Insdel RLDCs}. Throughout this paper, we only consider the case where $\Sigma=\{0,1\}$.


Definition \ref{def:strongRLDC} has also been extended recently to the notion of {\em Relaxed Locally Correctable Codes (RLCCs)} by Gur, Ramnarayan, and Rothblum \cite{GurRR20}. RLDCs and RLCCs have been studied  in a sequence of exciting works, where new upper and lower bounds have emerged, and new applications to probabilistic proof systems have been discovered \cite{gur2019lower, ChiesaGS20, GurRR20, AsadiS21, GurL21}.

Surprisingly, Ben-Sasson \etal \cite{Ben-SassonGHSV06} construct strong RLDCs with $q=O(1)$ queries and  $m=n^{1+O(1/\sqrt{q})}$, and more recently Asadi and Shinkar \cite{AsadiS21} improve the bounds to $m=n^{1+O(1/q)}$, in stark contrast with the state-of-the-art constructions of standard LDCs. 
Gur and Lachish \cite{GurL21} show that these bounds are in fact tight, as for every $q\geq 2$, every weak $q$-query RLDC must have length  $m=n^{1+1/O(q^{2})}$ for non-adaptive decoders. 
 We remark that the lower bounds of \cite{GurL21} hold even when the decoder does not have perfect completeness and in particular valid message bits are decoded with success probability $2/3.$ 
 Dall'Agnon, Gur, and Lachish \cite{dall2021structural} further extend these bounds to the setting where the decoder is adaptive, with $m=n^{1+1/O(q^{2}\log^2 q)}.$ 
 
\subsection{Our results}

As discussed before, since the introduction of RLDCs, unlike standard LDCs, they displayed a behaviour amenable to nearly linear-size constructions, with almost matching upper and lower bounds. However, recently \cite{GurL21} conjecture that for $q=2$ queries, there is in fact an exponential lower bound, matching the bounds for standard LDCs. 
 
In this paper, our first contribution is a proof of their conjecture, namely to show that Hamming $2$-query RLDCs  require exponential length. In fact, our exponential lower bound for $q=2$ applies even to weak RLDCs, which only satisfy the first two properties (perfect completeness and relaxed decoding), and even for adaptive decoders.

\begin{restatable}{theorem}{twoqrldcmain} \label{thm:main-2qRLDC}
	Let $C \colon \set{0,1}^n \rightarrow \set{0,1}^m$ be a weak adaptive $(2,\delta,1/2+\eps)$-RLDC. Then $m = 2^{\Omega_{\delta,\eps}(n)}$.
\end{restatable}
Our results are the first exponential bounds for RLDCs. Furthermore, combined with the constructions with nearly linear codeword length for some constant number of queries \cite{Ben-SassonGHSV06, AsadiS21}, our results imply that RLDCs experience a ``phase transition''-type phenomena, where the codeword length drops from being exponential at $q=2$ queries to being almost linear at $q=c$ queries for some constant $c> 2$. In particular, this also implies that there is a query number $q$ where the codeword length drops from being super-polynomial at $q$ to being polynomial at $q+1$. Finding this exact threshold query complexity is an intriguing open question.    


As our second contribution, we introduce and study the notion of RLDCs  correcting {\em insertions and deletions}, namely  Insdel RLDCs. 
The non-relaxed variants of Insdel LDCs  were first introduced in \cite{Ostrovsky-InsdelLDC-Compiler}, and were further studied in \cite{BlockBGKZ20,ChengLZ20,block2021private}. Local decoding in the Insdel setting is  motivated in DNA storage \cite{Olgica17}, and in particular  \cite{Banaletal-nature2021} show recent advances in bio-technological aspects of random access to data in these precise settings.

In \cite{Ostrovsky-InsdelLDC-Compiler,BlockBGKZ20}, the authors give Hamming to Insdel reductions which transform any Hamming LDC into an Insdel LDC with rate reduced by a constant multiplicative factor, and locality increased by a  $\polylog(m)$ multiplicative factor. Unfortunately, these compilers do not imply constant-query Insdel LDCs, whose existence is still an open question.

 The results of \cite{blocki2021exponential} show  strong lower bounds on the length of constant-query Insdel LDCs. In particular, they show that  linear Insdel LDCs with $2$  queries do no exist,  general Insdel LDCs for $q= 3$ queries must have  $m=\exp(\Omega(\sqrt{n}))$, and for $q\geq 4$ they must have $m=\exp(n^{\Omega(1/q)}).$

 In this work we continue the study of locally decodable codes in insertion and deletion channels by proving the first upper and lower bounds regarding the relaxed variants of Insdel LDCs. 
 We first consider strong Insdel RLDCs, which satisfy all three properties of Definition \ref{def:strongRLDC} and where the notion of distance is now that of relative edit distance. We adapt and extend the results of \cite{blocki2021exponential} to establish strong lower bounds on the codeword length of strong Insdel RLDCs. In particular, we prove that $m= \exp(n^{\Omega(1/q)})$ for any strong Insdel RLDC with locality $q$.

 \begin{restatable}{theorem}{strirldcmain} \label{thm:main-sinsdel RLDC}
	Let $C \colon \set{0,1}^n \rightarrow \set{0,1}^m$ be a non-adaptive strong $(q,\delta,1/2+\beta,\rho)$-Insdel RLDC where $\beta>0$. Then for every $q\ge 2$ there is a constant $c_1=c_1(q,\delta,\beta,\rho)$ such that  
	\begin{align*}
	    m = \exp\tp{c_1\cdot n^{\Omega_{\rho}(\beta^2/q)}}.
	\end{align*}
	Furthermore, the same bound holds even if $C$ does not have perfect completeness. If $C$ has an adaptive decoder, the same bound holds with $\beta$ replaced by $\beta/2^{q-1}$. Formally, there exists a constant $c_2=c_1(q,\delta,\beta/2^{q-1},\rho)$ such that
	\begin{align*}
	    m = \exp\tp{c_2\cdot n^{\Omega_{\rho}(\beta^2/(q2^{2q}))}}.
	\end{align*}
\end{restatable} 

Our reduction shown in the proof of \cref{thm:main-2qRLDC},  together with the impossibility results of standard {\em linear} or {\em affine} 2-query Insdel LDCs from \cite{blocki2021exponential} show a further impossibility result for linear and for affine $2$-query Insdel RLDCs (see remarks before \cref{cor:linear-2qIRLDC}). A linear code of length $m$ is defined over a finite field $\F$ and it is a linear subspace of the vector space $\F^m$, while an affine code is an affine subspace of $\F^m$.

 We then consider {\em weak} Insdel RLDCs that only satisfy the first two properties (perfect completeness and relaxed decoding). In contrast with \cref{thm:main-sinsdel RLDC}, we construct weak Insdel RLDCs with constant locality $q=O(1)$ and length $m=n^{1+\gamma}$ for some constant $\gamma \in (0, 1)$. To the best of our knowledge, this is the first positive result in the constant-query regime and the Insdel setting. However, the existence of a constant-query standard Insdel LDC (or even a constant-query strong Insdel RLDC) with any rate remains an open question. Finally, it is easy to see that our exponential lower bound for weak Hamming RLDCs with locality $q=2$ still applies in the Insdel setting, since Insdel errors are more general than Hamming error. Thus, in the Insdel setting we discover the same ``phase transition''-type phenomena as for Hamming RLDCs.

\begin{restatable}{theorem}{weakirldcmain} \label{thm:main-rildc-construction}
	For any $\gamma>0$ and $\eps \in (0,1/2)$, there exist constants $\delta \in (0,1/2)$ and $q=q(\delta,\eps,\gamma)$, and non-adaptive weak $(q,\delta,1/2+\eps)$-Insdel RLDCs $C\colon\set{0,1}^n\rightarrow \set{0,1}^m$ with $m=O(n^{1+\gamma})$. 
\end{restatable}

 We remark that in the Hamming setting, \cite{Ben-SassonGHSV06} shows that the first two properties of \cref{def:strongRLDC} imply the third property for codes with constant query complexity and which can withstand a constant fraction of errors.
 Our results demonstrate that, in general, unlike in the Hamming case, the first two properties do not imply the third property for Insdel RLDCs from \cref{def:strongRLDC}. 
 Indeed, while for strong Insdel RLDCs we have $m= \exp(n^{\Omega(1/q)})$  for codes of locality $q$, there exists $q=O(1)$ for which we have constructions of weak Insdel RLDCs with $m=n^{1+\gamma}.$ This observation suggests that there are significant differences between Hamming RLDCs and Insdel RLDCs. 

 We note that our construction of weak Insdel RLDCs can be modified to obtain strong Insdel Relaxed Locally Correctable Codes (Insdel RLCCs). Informally, an Insdel RLCC is a code for which codeword entries can be decoded to the correct value or $\bot$ with high probability,   even in the presence of insdel errors. 
 The formal definition of RLCC is given in \cref{sec:rLCC} (see \cref{def:rLCC}). We have the following corollary.
 
 \begin{restatable}{corollary}{strongirlcc}\label{thm:main-rilcc-construction}
 For any $\gamma>0$ and $\eps \in (0,1/2)$, there exist constants $\delta \in (0,1/2)$ and $q=q(\delta,\eps,\gamma)$, and non-adaptive strong $(q,\delta,1/2 + \eps, 1/2 )$-Insdel RLCCs $C\colon \set{0,1}^n\rightarrow \set{0,1}^m$ with $m=O(n^{1+\gamma})$. 
 \end{restatable}


\subsection{Overview of techniques}

\subsubsection{Exponential Lower Bound for Weak Hamming RLDCs with \texorpdfstring{$q=2$}{q=2}}
To simplify the presentation, we assume a non-adaptive decoder in this overview. While the exact same arguments do not directly apply to adaptive decoders\footnote{\label{foot:kt-obs}For standard LDCs Katz and Trevisan \cite{KatzT00} observed that an adaptive decoder could be converted into a non-adaptive decoder by randomly guessing the output $y_j$ of the first query $j$ to learn the second query $k$. Now we non-adaptively query the received codeword for both $y_j$ and $y_k$. If our guess for $y_j$ was correct then we continue simulating the adaptive decoder. Otherwise, we simply guess the output $x_i$. If the adaptive decoder succeeds with probability at least $p \geq 1/2 + \epsilon$ then the non-adaptive decoder succeeds with probability $p' \geq 1/4 + p/2 \geq 1/2 + \epsilon/2$. Unfortunately, this reduction does not preserve perfect completeness as required by our proofs for relaxed $2$-query Hamming RLDCs i.e., if $p=1$ then $p' = 3/4$.}, with a bit more care they can be adapted to work in those settings. 

At a high level we prove our lower bound by transforming any non-adaptive $2$-query weak Hamming RLDC for messages of length $n$ and $\delta$ fraction of errors into a standard $2$-query Hamming LDC for messages of length $n'=\Omega(n)$, with slightly reduced error tolerance of $\delta/2$. Kerenidis and de Wolf \cite{KerenidisW04} proved that any $2$-query Hamming LDC for messages of length $n$ must have codeword length $m = \exp(\Omega(n))$. Combining this result with our transformation, it immediately follows that any $2$-query weak Hamming RLDC must also have codeword length $m = \exp(\Omega(n))$. While our transformation does not need the third property (success rate) of a strong RLDC, we crucially rely on the property of {\em perfect completeness}, and that the decoder only makes $q=2$ queries.
 
Let $C\colon\set{0,1}^n \rightarrow \set{0,1}^m$ be a weak $(2,\delta,1/2+\eps)$-RLDC. For simplicity (and without loss of generality), let us assume the decoder $\Dec$ works as follows. For message $x$ and input $i \in [n]$, the decoder non-adaptively makes 2 random queries $j ,k \in [m]$, and outputs $f_{j,k}^{i}(y_j, y_k) \in \set{0,1,\perp}$, where $y_j, y_k$ are answers to the queries from a received word $y$, and $f_{j,k}^{i} \colon \zo^2 \rightarrow \set{0,1,\perp}$ is a deterministic function. When there is no error, we have $y_j = C(x)_j$ and $y_k = C(x)_k$. 

We present the main ideas below, and refer the readers to \cref{sec:2qrldc} for full details.

\paragraph{Fixable codeword bits.} The starting point of our proof is to take a closer look at those functions $f_{j,k}^{i}$ with $\perp$ entries in their truth tables. It turns out that when $f_{j,k}^{i}$ has at least one $\perp$ entry in the truth table, $C(x)_j$ can be fixed to a constant by setting either $x_i=0$ or $x_i=1$, and same for $C(x)_k$. To see this, note that the property of perfect completeness forces $f_{j,k}^{i}$ to be $0$ or $1$ whenever $x_i=0$ or $x_i=1$ and there is no error. Thus if neither $x_i=0$ nor $x_i=1$ fixes $C(x)_j$, then there must be two entries of $0$ and two entries of $1$ in the truth table of $f_{j,k}^{i}$, which leaves no space for $\perp$ (see \cref{clm:bot-fix}). Thus, when there is at least one $\perp$ entry in the truth table of $f_{j,k}^{i}$, we say that $C(x)_j$ and $C(x)_k$ are \emph{fixable} by $x_i$.

This motivates the definition of the set $S_i$, which contains all indices $j \in [m]$ such that the codeword bits $C(x)_j$ are fixable by $x_i$; and the definition of $T_j$, the set of all indices $i \in [n]$ such that $C(x)_j$ is fixable by the message bits $x_i$.\ It is also natural to pay special attention to queries $j,k$ that are not both contained in $S_i$, since in this case the function $f_{j,k}^{i}$ never outputs $\perp$.

\paragraph{The query structure.} In general, a query set $\set{j,k}$ falls into one of the following three cases: (1) both $j,k$ lie inside $S_i$; (2) both $j,k$ lie outside of $S_i$; (3) one of them lies inside $S_i$ and the other lies outside of $S_i$. It turns out that case (3) essentially never occurs for a decoder with perfect completeness. The reason is that when, say, $j \in S_i$ and $k \notin S_i$, one can effectively pin down every entry in the truth table of $f_{j,k}^{i}$ by using the perfect completeness property, and observe that the output of $f_{j,k}^{i}$ does not depend on $y_k$ at all (see \cref{clm:fixable-or-useless}). Thus in this case we can equivalently view the decoder as only querying $y_j$ where $j \in S_i$, which leads us back to case (1). In what follows, we denote by $E_1$ the event that case (1) occurs, and by $E_2$ the event that case (2) occurs.

\paragraph{The transformation by polarizing conditional success probabilities.} We now give a high level description of our transformation from a weak RLDC to a standard LDC. Let $y$ be a string which contains at most $\delta m/2$ errors from the codeword $C(x)$. We have established that the success probability of the weak RLDC decoder on $y$ is an average of two conditional probabilities
\begin{align*}
	\Pr[\Dec(i,y) \in \set{x_i, \perp}] = p_1 \cdot \Pr[\Dec(i,y) \in \set{x_i, \perp} \mid E_1] + p_2 \cdot \Pr[\Dec(i,y) \in \set{x_i, \perp} \mid E_2],
\end{align*}
where $p_1 = \Pr[E_1]$ and $p_2 = \Pr[E_2]$. Let us assume for the moment that $S_i$ has a small size, e.g., $|S_i|\le \delta m/2$. The idea in this step is to introduce additional errors to the $S_i$-portion of $y$, in a way that drops the conditional success probability $\Pr[\Dec(i,y) \in \set{x_i, \perp} \mid E_1]$ to 0 (see \cref{lem:conditional-zero}). In particular, we modify the bits in $S_i$ to make it consistent with the encoding of any message $\hat{x}$ with $\hat{x}_i=1-x_i$. Perfect completeness thus forces the decoder to output $1-x_i$ conditioned on $E_1$. Note that we have introduced at most $\delta m/2 + |S_i| \le \delta m$ errors in total, meaning that the decoder should still have an overall success probability of $1/2+\eps$. Furthermore, now the conditional probability $\Pr[\Dec(i,y) \in \set{x_i, \perp} \mid E_2]$ takes all credits for the overall success probability. Combined with the observation that $\Dec$ never outputs $\perp$ given $E_2$, this suggests the following natural way to decode $x_i$ in the sense of a standard LDC: sample queries $j, k$ according to the conditional probability given $E_2$ (i.e., both $j,k$ lie outside $S_i$) and output $f_{j,k}^{i}(y_j, y_k)$. This gives a decoding algorithm for standard LDC, with success probability $1/2+\eps$ and error tolerance $\delta m/2$ (see \cref{lem:LDC-reduction}), modulo the assumption that $|S_i|\le \delta m/2$.

\paragraph{Upper bounding \texorpdfstring{$|S_i|$}{|S\_i|}.} The final piece in our transformation from weak RLDC to standard LDC is to address the assumption that $|S_i|\le \delta m/2$. This turns out to be not true in general, but it would still suffice to prove that $|S_i| \leq \delta m/2$ for $n' = \Omega(n)$ of the message bits $i$. If we could show that $|T_j|$ is small for most $j \in [m]$, then a double counting argument shows that $|S_i|$ is small for most $i \in [n]$. Unfortunately, if we had  $C(x)_j = \bigwedge_{i=1}^n x_i$ for $m/2$ of the codeword bits $j$ then we also have $|T_j| = n$ for $m/2$ codeword bits and $|S_i| \geq m/2 \geq \delta m/2$ for all message bits $i \in [n]$. We address this challenge by first arguing that any weak RLDC for $n$-bit messages can be transformed into another weak RLDC for $\Omega(n)$-bit messages for which we have $|T_j|  \le 3\ln(8/\delta) $ for all but $\delta m/4$ codeword bits. The transformation works by fixing some of the message bits and then eliminating codeword bits that are fixed to constants. Intuitively, if some $C(x)_j$ is fixable by many message bits, it will have very low entropy (e.g., $C(x)_j$ is the AND of many message bits) and hence contain very little information and can (likely) be eliminated. We make this intuition rigorous through the idea of random restriction: for each $i \in [n]$, we fix $x_i=0$, $x_i=1$, or leave $x_i$ free, each with probability $1/3$. The probability that $C(x)_j$ is not fixed to a constant is at most $(1-1/3)^{|T_j|}\le \delta/8$, provided that $|T_j| \ge 3\ln(8/\delta)$. After eliminating codeword bits that are fixed to constants, we show that with probability at least $1/2$ at most $\delta m/4$ codeword bits $C(x)_j$ with $|T_j| \ge 3\ln(8/\delta)$ survived\footnote{We are oversimplifying a bit for ease of presentation. In particular, the random restriction process may cause a codeword bit $C(x)_j$ to be fixable by a new message bit $x_i$ that did not belong to $T_j$ before the restriction -- We thank an anonymous reviewer for pointing this out to us.
Nevertheless, for our purpose it is sufficient to eliminate codeword bits that initially have a large $|T_j|$. See the formal proof for more details.}. Note that with high probability the random restriction leaves at least $n/6$ message bits free. Thus, there must exist a restriction which leaves at least $n/6$ message bits free ensuring that $|T_j| \ge 3\ln(8/\delta)$ for at most $\delta m/4$ of the remaining codeword bits $C(x)_j$. We can now apply the double counting argument to conclude that $|S_i| \le \delta m/2$ for $\Omega(n)$ message bits, completing the transformation.


\paragraph{Adaptive decoders} For possibly adaptive decoders, we are going to follow the same proof strategy. The new idea and main difference is that we focus on the first query made by the decoder, which is always non-adaptive. We manage to show that the first query determines a similar query structure, which is the key to the transformation to a standard LDC. More details can be found in \Cref{subsec:adaptive-2qRLDC}.

\subsubsection{Lower Bounds for Strong Insdel RLDCs}
We recall that a strong Insdel RLDC $C$ is a weak Insdel RLDC which satisfies an additional property: for every $x \in \zo^n$ and $y \in \zo^{m'}$ such that $\ED(C(x), y) \le \delta\cdot 2m$, 
there exists a set $I_y \subseteq [n]$ of size $|I_y|\ge \rho n$ such that for every $i \in I_y$, we have $\Pr[\Dec(i,y) = x_i] \ge \alpha$. In other words, for $\rho$-fraction of the message bits, the decoder can correctly recover them with high probability, just like in a standard Insdel LDC. 
Towards obtaining a lower bound on the codeword length $m$, a natural idea would be to view $C$ as a standard Insdel LDC just for that $\rho$-fraction of message bits, and then apply the exponential lower bound for standard Insdel LDCs from \cite{blocki2021exponential}. 
This idea would succeed if the message bits correctly decoded with high probability were the same for all potential corrupted codewords $y$. 
However, it could be the case that $i \in I_y$ for some strings $y$, whereas $i \notin I_{y'}$ for other strings $y'$. Indeed, allowing the set $I_y$ to depend on $y$ is the main reason why very short constant-query Hamming RLDCs exist.  

We further develop this observation to obtain our lower bound. 
We use an averaging argument to show the existence of a \emph{corruption-independent} set $I$ of message bits with $|I|=\Omega(n)$, which the decoder can recover with high probability. To this end, we need to open the ``black box'' of the lower bound result of Blocki \etal \cite{blocki2021exponential}. The proof in \cite{blocki2021exponential} starts by constructing an error distribution $\+E$ with several nice properties, and deduce the exponential lower bound based solely on the fact that the Insdel LDC should, on average (i.e., for a uniformly random message $x$), correctly recover each bit with high probability under $\+E$ (see \cref{thm:lb-channel}). One of the nice properties of $\+E$ is that it is oblivious to the decoding algorithm $\Dec$. Therefore, it makes sense to consider the average success rate against $\+E$, i.e., $\Pr[\Dec(i,y)=x_i]$, where $i \in [n]$ is a uniformly random index, $x \in \zo^n$ is a uniformly random string, and $y$ is a random string obtained by applying $\+E$ to $C(x)$. By replacing $\perp$ with a uniformly random bit in the output of $\Dec$, the average success rate is at least $\rho\alpha+(1-\rho)\alpha/2=(1+\rho)\alpha/2$, since there is a $\rho$-fraction of indices for which $\Dec$ can correctly recover with probability $\alpha$, and for the remaining $(1-\rho)$-fraction of indices the random guess provides an additional success rate of at least $\alpha/2$. Assuming $\alpha$ is sufficiently close to 1, which we can achieve by repeating the queries independently for a constant number of times and doing something similar to a majority vote, the average success rate against $\+E$ is strictly above $1/2$. Therefore, there exist a constant fraction of indices for which the success rate against $\+E$ is still strictly above $1/2$, and the number of queries remains a constant. This is sufficient for the purpose of applying the argument in~\cite{blocki2021exponential} to get an exponential lower bound.
Full details appear in \cref{sec:srldc}.

\subsubsection{Constant-Query Weak Insdel RLDC}
Our construction of a constant query weak Insdel RLDC uses code concatenation and two building blocks: a weak Hamming RLDC (as the outer code) with constant query complexity, constant error-tolerance, and codeword length $k = O(n^{1+\gamma})$ for any $\gamma > 0$ \cite{Ben-SassonGHSV06}, and the Schulman-Zuckerman \cite{SchZuc99} (from now on denoted by SZ) Insdel codes\footnote{In particular, these are classical/non-local codes.} (as the inner code).
We let $\Cout \colon \zo^n \rightarrow \zo^k$ and $\Cin \colon [k] \times \zo \rightarrow \zo^t$ denote the outer and inner codes, respectively.
Our final concatenation code $C$ will have codewords in $\zo^m$ for some $m$ (to be determined shortly), will have constant query complexity, and will tolerate a constant fraction of Insdel errors.

\paragraph{Code construction.}
Given a message $x \in \zo^n$, we first apply the outer code to obtain a Hamming codeword $y = y_1 \circ \cdots \circ y_k = \Cout(x)$ of length $k$, where each $y_i \in \{0,1\}$ denotes a single bit of the codeword. 
Then for each index $i$, we compute $c_i = \Cin(i, y_i) \in \{0,1\}^t$ as the encoding of the message $(i, y_i)$ via the inner code. 
Finally, we output the codeword $C(x) \coloneqq c_1 \circ 0^t \circ c_2 \circ \cdots \circ  0^t \circ c_k$, where $0^t$ denotes a string of $t$ zeros (which we later refer to as a buffer). 
Note that the inner code is a constant-rate code, i.e., $t = O(\log(k))$, and has constant error-tolerance $\din \in (0,1/2)$.
Thus, the final codeword has length $m \coloneqq (2t-1)k = O(k \log(k))$ bits. 
For any constant $\gamma > 0$ we have a constant query outer code with length $k = O(n^{1+\gamma})$. 
Plugging this into our construction we have codeword length $ m = O( n^{1+\gamma} \log n)$ which is $O(n^{1+\gamma'})$ for any constant $\gamma' > \gamma$. 

\paragraph{Decoding algorithm: intuition and challenges.}
Intuitively, our relaxed decoder will simulate the outer decoder. 
When the outer decoder requests $y_i$, the natural approach would be to find and decode the block $c_i$ to obtain $(i,y_i)$. 
There are two challenges in this approach. 
First, if there were insertions or deletions, then we do not know where the block $c_i$ is located; moreover, searching for this block can potentially blow-up the query complexity by a multiplicative $\polylog(m)$ factor \cite{Ostrovsky-InsdelLDC-Compiler,BlockBGKZ20}. 
Second, even if we knew where $c_i$ were located, because $t = O(\log k)$ and we want the decoder to have constant locality, we cannot afford to recover the entire block $c_i$. 

We address the first challenge by attempting to locate block $c_i$ under the optimistic assumption that there are no corruptions. 
If we detect any corruptions, then we may immediately abort and output $\bot$ since our goal is only to obtain a weak Insdel RLDC. 
Assuming that there were no corruptions, we know exactly where the block $c_i$ is located, and we know that $c_i$ can only take on two possible values: it is either the inner encoding of $(i,0)$ or the inner encoding of $(i,1)$. 
If we find anything inconsistent with the inner encoding of either $(i,0)$ or $(i,1)$, then we can immediately output $\bot$.

Checking consistency with the inner encodings of $(i,0)$ and $(i,1)$ is exactly how we address the second challenge.
In place of reading the entire block $c_i$, we instead only need to determine whether (1) $c_i$ is (close to) the inner encoding of $(i,0)$, (2) $c_i$ is (close to) the inner encoding of $(i,1)$, or (3) $c_i$ is not close to either string.
In either case (1) or case (2), we simply output the appropriate bit, and in case (3), we simply output $\bot$.
Thus, our Insdel RLDC decoder simulates the outer decoder. 
Whenever the outer decoder request $y_i$, we determine the expected location for $c_i$, randomly sub-sample a constant number of indices from this block and compare with the inner encodings of $(i,0)$ and $(i,1)$ at the corresponding indices. 
To ensure perfect completeness, we always ensure that {\em at least one} of the sub-sampled indices is for a bit where the inner encodings of $(i,0)$ and $(i,1)$ differ. 
If there are no corruptions, then whenever the simulated outer decoder requests $y_i$ we will always respond with the correct bit. 
Perfect completeness of our Insdel RLDC now follows immediately from the perfect completeness of the outer decoder. 
Choosing a constant number of indices to sub-sample ensures that the locality of our weak Insdel RLDC decoder is a constant multiplicative factor larger than the outer decoder, which gives our Insdel RLDC decoder constant locality overall.

\paragraph{Analysis of the decoding algorithm.} 
The main technical challenge is proving that our Insdel RLDC still satisfies the second condition of \cref{def:strongRLDC}, when the received word is not a correct encoding of the message $x$. 
Recall that $c_i = \Cin(i,y_i)$, and suppose $\tildec_i \neq c_i$ is the block of the received word that we are going to check for consistency with the inner encodings of $(i,0)$ and $(i,1)$.
Then, the analysis of our decoder falls into three cases.
In the first case, if $\tildec_i$ is not too corrupted (\ie, $\ED(\tildec_i, c_i)$ is not too large), then we can argue that the decoder outputs the correct bit $y_i$ or $\bot$ with good probability.
In the second case, if $\tildec_i$ has high edit distance from both $\Cin(i,0)$ and $\Cin(i,1)$, then we can argue that the decoder outputs $\bot$ with good probability.
The third case is the most difficult case, which we describe as ``dangerous''. 
We say that the block $\tildec_i$ is \emph{dangerous} if the edit distance between $\tildec_i$ and $\Cin(i,1-y_i)$ is not too large; \ie, $\tildec_i$ is close to the encoding of the opposite bit $1-y_i$.

The key insight to our decoding algorithm is that as long as the number of dangerous blocks $\tildec_i$ is upper bounded, then we can argue the overall probability that our decoder outputs $y_i$ or $\bot$ satisfies the relaxed decoding condition of \cref{def:strongRLDC}.
Intuitively, we can we think of our weak Insdel RLDC decoder as running the outer decoder on a string $\tilde{y} = \tilde{y}_1 \circ \ldots \circ \tilde{y}_k$, where each $\tilde{y}_i \in \{0,1,\bot\}$ and the outer decoder has been modified to output $\bot$ whenever it queries for $y_i$ and receives $\bot$. 
Observe that if $\dout$ is the error-tolerance of the outer decoder, then as long as the set $\left| \{ i ~:~\tildey_i \neq \bot \wedge \tildey_i \neq y_i \} \right| \leq \dout k$, the modified outer decoder, on input $j \in [n]$, will output either the correct value $x_j$ or $\bot$ with high probability 
(for appropriate choices of parameters). 
Intuitively, if a block is ``dangerous'' then we can view $\tildey_i = 1- y_i$, and otherwise we have $\tildey_i \in \{y_i, \bot\}$ with reasonably high probability. 
Thus, as long as the number of ``dangerous'' block is at most $\dout k/2$, then our relaxed Insdel decoder will satisfy the second property of \cref{def:strongRLDC} and output either $x_j$ or $\bot$ with high probability 
for any $j \in [n]$.

\paragraph{Upper bounding the number of dangerous blocks.}
To upper bound the number of ``dangerous'' blocks we utilize a matching argument based on the longest common sub-sequence (LCS) between the original codeword and the received (corrupted) word. 
Our matching argument utilizes a key feature of the SZ Insdel code. 
In particular, the Hamming weight (\ie, number of non-zero symbols) of every substring $c'$ of an SZ codeword is at least $\floor{|c'|/2}$. 
This ensures that the buffers $0^t$ cannot be matched with large portions of any SZ codeword. 
We additionally leverage a key lemma (\cref{lem:self-nonsimilarity}) which states that the edit distance between the codeword $\Cin(i,1-y_i)$ and \emph{any} substring of length less than $2t$ of the uncorrupted codeword $C(x)$ has relative edit distance at least $\din/2$.
We use these two properties, along with key facts about the LCS matching, to yield an upper bound on the number of dangerous blocks, completing the analysis of our decoder.


\paragraph{Extension to relaxed locally correctable codes for insdel errors.}
Our construction also yields a strong Insdel Relaxed Locally Correctable Code (RLCC) with constant locality if the outer code is a weak Hamming RLCC. 
First, observe that bits of the codeword corresponding to the $0^t$ buffers are very easy to predict without even making any queries to the corrupted codeword. 
Thus, if we are asked to recover the $j$'th bit of the codeword and $j$ corresponds to a buffer $0^t$, we can simply return $0$ without making any queries to the received word. 
Otherwise, if we are asked to recover the $j$'th bit of the codeword and $j$ corresponds to block $c_i$, we can simulate the Hamming RLCC decoder (as above) on input $i$ to obtain $y_i$ (or $\bot$). 
Assuming that $y_i \in \{0,1\}$, we can compute the corresponding SZ encoding of $(i,y_i)$ and obtain the original value of the block $c_i$ and then recover the $j$'th bit of the original codeword.
The analysis of the RLCC decoder is analogous to the RLDC decoder.
See \cref{sec:wrLDC} for full details on both our weak Insdel RLDC and strong Insdel RLCC constructions.
\begin{remark}
    The ``adaptiveness'' of our constructed Insdel RLDC/RLCC decoder is identical to that of the outer Hamming RLDC/RLCC decoder.
    In particular, the weak Hamming RLDC of Ben-Sasson \etal \cite{Ben-SassonGHSV06} has a non-adaptive decoder, making our final decoder non-adaptive as well.
    Similarly, we use a weak Hamming RLCC due to Asadi and Shinkar \cite{AsadiS21} for our Insdel LCC, which is also a non-adaptive decoder.
\end{remark}

\section{Open Questions}\label{sec:openproblems}

\paragraph{Exact ``phase-transition'' thresholds.} Our results show that both in the Hamming and Insdel setting there is a constant $q$ such that every $q$-query RLDC requires super-polynomial codeword length, while there exists a $(q+1)$-query RLDC of polynomial codeword length. 
Finding the precise $q$ remains an intriguing open question. Further, a more refined understanding of codeword length for RLDCs making $3, 4, 5$ queries is another important question, which has lead to much progress in the understanding of the LDC variants.


\paragraph{Constant-query strong Insdel RLDCs/RLCCs.} While we do construct the first weak RLDCs in the Insdel setting, the drawback of our constructions is the fact that our codes do not satisfy the third property of Definition \ref{def:strongRLDC}. Building strong Insdel RLDCs remains an open question. We note that our lower bounds imply that for a constant number of queries, such codes (if they exist) must have exponential codeword length.



\paragraph{Applications of local Insdel codes.} As previously mentioned, Hamming LDCs/RLDCs have so far found many applications such as private information retrieval, probabilistically checkable proofs, self-correction,  fault-tolerant circuits,  hardness amplification, and data structures. Are there analogous or new applications of the Insdel variants in the broader computing area?

\paragraph{Lower bounds for Hamming RLDCs/LDCs} Our $2$-query lower bound for Hamming RLDCs crucially uses the perfect completeness property of the decoder. An immediate question is whether the bound still holds if we allow the decoder to have imperfect completeness. We also note that the argument in our exponential lower bounds for $2$-query Hamming RLDCs fail to hold for alphabets other than the binary alphabet, and we leave the extension to larger alphabet sizes as an open problem. 
Another related question is to understand if one can leverage perfect completeness and/or random restrictions to obtain improved lower bounds for $q\geq 3$-query standard Hamming LDCs. Perfect completeness has been explicitly used before to show exponential lower bounds for $2$-query LCCs \cite{bhattacharyya2017lower}.

\subsection{Further discussion about related work}


\paragraph{Insdel codes.} The study of error correcting codes for insertions and deletions was initiated by Levenstein \cite{Levenshtein_SPD66}. While  progress has been slow because constructing codes for insdel errors is strictly more challenging than for Hamming errors, strong interest in these codes lately has led to many exciting results  \cite{SchZuc99, Kiwi_expectedlength, guruswami2017deletion,HaeuplerS17, GuruswamiL18,HaeuplerSS18,HaeuplerS18,BrakensiekGZ18, ChengJLW18, ChengHLSW19, ChengJ0W19, GuruswamiL19,HaeuplerRS19,Haeupler19,  LiuTX20,GuruswamiHS20, ChengGHL21, ChengL21} (See also the excellent surveys of \cite{Sloane2002OnSC,Mercier2010ASO,Mitzenmachen-survey, haeupler2021synchronization}).

\paragraph{Insdel LDCs.} 
\cite{OPS07} gave private-key constructions of LDCs with  $m=\Theta(n)$ and locality $\polylog(n)$. \cite{BlockiKZ19} extended the construction from \cite{OPS07} to settings where the sender/decoder do not share randomness, but the adversarial channel is resource bounded. 
\cite{block2021private} applied the \cite{BlockBGKZ20} compiler to the private key Hamming LDC of \cite{OPS07} (resp. resource bounded LDCs of \cite{BlockiKZ19}) to obtain private key Insdel LDCs (resp. resource bounded Insdel LDCs) with constant rate and $\polylog(n)$ locality.

Insdel LDCs have also been recently studied in {\em computationally bounded channels}, introduced in \cite{Lipton94}. Such channels can perform a bounded number of adversarial errors, but do not have unlimited computational power as the general Hamming channels. Instead, such channels operate with bounded resources.
As expected, in many such  limited-resource settings one can construct codes with strictly better parameters than what can be done generally \cite{GopalanLD04,MicaliPSW05, Guruswami_Smith:2016, ShaltielS16}. LDCs in these channels under Hamming error were studied in \cite{OPS07, HemenwayO08, HemenwayOSW11, HemenwayOW15, BlockiGGZ19,BlockiKZ19}. \cite{block2021private} applied the \cite{BlockBGKZ20} compiler to the Hamming LDC of \cite{BlockiKZ19} to obtain a constant rate Insdel LDCs with $\polylog(n)$ locality for resource bounded channels.  The work of \cite{ChengLZ20}  proposes the notion of locally decodable codes with randomized encoding, in both the Hamming and edit distance regimes, and in the setting where the channel is oblivious to the encoded message, or the encoder and decoder share randomness. For edit error they obtain codes with $m=O(n)$ or $m= n \log n$ and $\polylog(n)$ query complexity. However, even in settings with shared randomness or where the channel is oblivious or resource bounded, there are no known constructions of Insdel LDCs with constant locality.
 
Locality in the study of insdel codes was also considered in   \cite{HaeuplerS18}, which constructs explicit synchronization strings that can be locally decoded.


\subsection{Organization}
The remainder of the paper is dedicated to proving all our results presented in \cref{sec:intro}. 
We give general preliminaries and recall some prior results used in our results in \cref{sec:prelims}.
We prove \cref{thm:main-2qRLDC} in \cref{sec:2qrldc}, prove \cref{thm:main-sinsdel RLDC} in \cref{sec:srldc}, and prove \cref{thm:main-rildc-construction,thm:main-rilcc-construction} in \cref{sec:wrLDC}.

\section{Preliminaries}\label{sec:prelims}
For natural number $n \in \NAT$, we let $[n] \coloneqq \{1,2,\dotsc, n\}$.
We let ``$\circ$'' denote the standard string concatenation operation.
For a string $x \in \zo^*$ of finite length, we let $|x|$ denote the length of $x$.
For $i \in [|x|]$, we let $x[i]$ denote the $i$-th bit of $x$. 
Furthermore, for $I \subseteq [|x|]$, we let $x[I]$ denote the subsequence $x[i_1] \circ x[i_2] \circ \cdots \circ x[i_\ell]$, where $i_j \in I$ and $\ell = |I|$.
For two strings $x, y \in \zo^n$ of length $n$, we let $\HAM(x,y)$ denote the \emph{Hamming Distance} between $x$ and $y$; \ie, $\HAM(x,y) \coloneqq \abs{\{ i \in [n] \colon x_i \neq y_i \}}$.
Similarly, we let $\ED(x,y)$ denote the \emph{Edit Distance} between $x$ and $y$; \ie, $\ED(x,y)$ is the minimum number of insertions and deletions needed to transform string $x$ into string $y$.
We often discuss the \emph{relative Hamming Distance} (resp., \emph{relative Edit Distance}) between $x$ and $y$, which is simply the Hamming Distance normalized by $n$, \ie, $\HAM(x,y)/n$ (resp., the Edit Distance normalized by $|x| + |y|$, \ie, $\ED(x,y)/(|x|+|y|)$).
Finally, the \emph{Hamming weight} of a string $x$ is the number of non-zero entries of $x$, which we denote as $\mathsf{wt}(x) \coloneqq |\{i \in [|x|] \colon x_i \neq 0\}|$.

For completeness, we recall the definition of a classical locally decodable code, or just a \emph{locally decodable code}.
\begin{definition}[Locally Decodable Codes]\label{def:LDC}
    A $(q, \delta, \alpha)$-Locally Decodable Code $C \colon \Sigma^n \rightarrow \Sigma^m$ is a code for which there exists a randomized decoder that makes at most $q$ queries to the received word $y$ and satisfies the following property: for every $i \in [n]$, if $y$ is such that $\dist(y, C(x)) \leq \delta$ for some unique $C(x)$, then the decoder, on input $i$, outputs $x_i$ with probability $\geq \alpha$.
    Here, the randomness is taken over the random coins of the decoder, and $\dist$ is a normalized metric.
    
    If $\dist$ is the relative Hamming distance, then we say that the code is a Hamming LDC; similarly, if $\dist$ is the relative edit distance, then we say that the code is an Insdel LDC.
\end{definition}

We recall the general $2$-query Hamming LDC lower bound  \cite{KerenidisW04, Ben-AroyaRW08}.
\begin{theorem}[\cite{KerenidisW04, Ben-AroyaRW08}]\label{thm:two-query-lb}
    For constants $\delta, \eps \in (0,1/2)$ there exists a constant $c = c(\delta, \eps) \in (0,1)$ such that if $C \colon \zo^n \rightarrow \zo^m$ is a $(2, \delta, 1/2+\eps)$ Hamming LDC then $m \geq 2^{cn-1}$.
\end{theorem}

In our weak Insdel RLDC construction, we utilize a weak Hamming RLDC due to \cite{Ben-SassonGHSV06}.
\begin{lemma}[\cite{Ben-SassonGHSV06}]\label{lem:weak-ham-rldc}
    For constants $\eps, \delta \in (0,1/2)$ and $\gamma \in (0,1)$, there exists a constant $q = O_{\delta,\eps}(1/\gamma^2)$ and a weak $(q,\delta, 1/2+\eps)$-Hamming RLDC $C \colon \zo^n \rightarrow \zo^m$ with $m = O(n^{1+\gamma})$.
    Moreover, the decoder of this code is non-adaptive.
\end{lemma}

Our construction additionally utilizes the well-known Schulman-Zuckerman Insdel codes \cite{SchZuc99}.
\begin{lemma}[Schulman-Zuckerman (SZ) Code \cite{SchZuc99}]\label{lem:sz-code}
    There exists constants $\beta \geq 1$ and $\delta > 0$ such that for large enough values of $t > 0$, there exists a code $C \colon \zo^t \rightarrow \zo^{\beta t}$ capable of decoding from $\delta$-fraction of Insdel errors and the additional property that for every $x \in \zo^t$ and $y = C(x)$, every substring $y'$ of $y$ with length at least $2$ has Hamming weight $\geq \floor{|y'|/2}$.
\end{lemma}

Our strong Insdel RLCC construction relies on a weak Hamming RLCC.
We utilize the following weak Hamming RLCC implicit in \cite{AsadiS21}.
\begin{lemma}[Implied by Theorem 1 of \cite{AsadiS21}]
\label{lem:rlcc}
    For every sufficiently large $q\in \NAT$ and $\eps\in (0,1/2)$, there is a constant $\delta$ such that there exists a weak $(q, \delta, 1/2 + \eps)$-relaxed Hamming Locally Correctable Code $C\colon \set{0,1}^n\rightarrow \set{0,1}^m$ with $m = n^{1+O(1/q)}$. 
    Moreover, the decoder of this code is non-adaptive.
\end{lemma}

\section{Lower Bounds for 2-Query Hamming RLDCs} \label{sec:2qrldc}

We prove \Cref{thm:main-2qRLDC} in this section. As a reminder, a weak $(q,\delta,\alpha)$-RLDC satisfies the first two conditions in \cref{def:strongRLDC}, and non-adaptive means the decoder makes queries according to a distribution which is independent of the received string $y$. Here we are interested in the case $q=2$ and $\alpha=1/2+\eps$.

To avoid overloading first-time readers with heavy notations, we first present a proof of the lower bound for \emph{non-adaptive} decoders, i.e., decoders with a query distribution independent of the received string. This proof will be easier to follow, while the crucial ideas behind it remain the same. The proof for the most general case is presented in the last subsection, with an emphasis on the nuances in dealing with adaptivity.

\subsection{A Warm-up: the lower bound for non-adaptive decoders}
In the following, we fix a relaxed decoder $\Dec$ for $C$. In this subsection, we assume that $\Dec$ is non-adaptive, and that it has the first two properties specified in \cref{def:strongRLDC}. To avoid technical details, we also assume $\Dec$ always makes exactly 2 queries (otherwise add dummy queries to make the query count exactly 2).


Given an index $i \in [n]$ and queries $j,k$ made by $\Dec(i,\cdot)$, in the most general setting the output could be a random variable which depends on $i$ and $y_j$, $y_k$, where $y_j$, $y_k$ are the answers to queries $j$, $k$, respectively. An equivalent view is that the decoder picks a random function $f$ according to some distribution and outputs $f(y_j, y_k)$. Let $\mathtt{DF}^i_{j,k}$ be the set of all decoding functions $f \colon \set{0,1}^2 \rightarrow \set{0,1,\perp}$ which are selected by $\Dec(i,\cdot)$ with non-zero probability when querying $j,k$. We partition the queries into the following two sets
\begin{align*}
F_{i}^{0} &\coloneqq \set{ \set{j,k} \subseteq [m] \colon \forall f \in \mathtt{DF}^i_{j,k} \textup{ the truth table of $f$ contains no ``$\perp$''}}, \\
F_{i}^{\ge 1} &\coloneqq \set{ \set{j,k} \subseteq [m] \colon \exists f \in \mathtt{DF}^i_{j,k} \textup{ the truth table of $f$ contains at least 1 ``$\perp$''}}.
\end{align*}

\paragraph{Notation.} Given a string $w \in \zo^m$ and a subset $S\subseteq [m]$, we denote $w[S]\coloneqq (w_i)_{i\in S} \in \zo^{|S|}$. Given a Boolean function $f \colon \set{0,1}^{n} \rightarrow \set{0,1}$, and $\sigma \in \set{0,1}$, we write $f\restriction_{x_i = \sigma}$ to denote the restriction of $f$ to the domain $\set{\*x \in \set{0,1}^n \colon x_i = \sigma}$. For a sequence of restrictions, we simply write $f\restriction_{(x_{j_1}, \dots, x_{j_k})=(\sigma_1,\dots,\sigma_k)}$, or $f_{J|\sigma}$ where $J=[n]\setminus\set{j_1,\dots,j_k}$ and $\sigma=(\sigma_1,\dots,\sigma_k)$. Note that $f_{J|\sigma}$ is a Boolean function over the domain $\zo^{J}$. 

We will identify the encoding function of $C$ as a collection of $m$ Boolean functions
\begin{align*} 
\+C \coloneqq \set{C_1, \dots, C_m \colon \forall j \in [m], C_j \colon \set{0,1}^n \rightarrow \set{0,1}}.
\end{align*}
Namely, $C(x)=(C_1(x), C_2(x), \dots, C_m(x))$ for all $x \in \zo^n$.

For $j \in [m]$, we say $C_j$ is \emph{fixable} by $x_i$ if at least one of the restrictions $C_j\restriction_{x_i=0}$ and $C_j\restriction_{x_i=1}$ is a constant function. Denote
\begin{align*}
S_i \coloneqq \set{j \in [m] \colon C_j\textup{ is fixable by }x_i}, \quad T_j \coloneqq \set{i \in [n] \colon C_j\textup{ is fixable by }x_i},
\end{align*}
and $w_{j} \coloneqq |T_j|$. Let 
\begin{align*}
W \coloneqq \set{j \in [m] \colon w_{j} \ge 3\ln(8/\delta) }.
\end{align*}
For $i \in [n]$ define the sets $S_{i,+} \coloneqq S_{i} \cap W$, and $S_{i,-} \coloneqq S_i \cap \overline{W}$.

Let $J \subseteq [n]$ and $\rho \in \set{0,1}^{\overline{J}}$. A code $C \colon \set{0,1}^n \rightarrow \set{0,1}^m$ restricted to $\*x_{\overline{J}} = \rho$, denoted by $C_{J|\rho}$, is specified by the following collection of Boolean functions
\begin{align*}
	\+C_{J|\rho} \coloneqq \set{C_j\restriction_{\*x_{\overline{J}}=\rho} \colon j \in [m], C_j\restriction_{\*x_{\overline{J}}=\rho}\textup{ is not a constant function}}.
\end{align*} 
Namely, we restrict each function $C_j$ in $\+C$ to $\*x_{\overline{J}}=\rho$, and eliminate those that have become constant functions. $C_{J|\rho}$ encodes $n'$-bit messages into $m'$-bit codewords, where $n'=|J|$ and $m' = \abs{\+C_{J|\rho}} \le m$. 

We note that the local decoder $\Dec$ for $C$ can also be used as a local decoder for $C_{J|\rho}$, while preserving all the parameters. This is because, $\Dec$ never needs to really read a codeword bit which has become a constant function under the restriction $J|\rho$. 

The lemma below will be useful later in the proof. It shows that a constant fraction of the message bits can be fixed so that most codeword bits $C_j$ with large $w_j$ become constants.

\begin{lemma} \label{lem:random-restriction}
There exist a set $J \subseteq [n]$ and assignments $\rho \in \set{0,1}^{\overline{J}}$ such that $|J| \ge n/6$, and $|W \setminus A| \le \delta m/4$, where $A \subseteq W$ collects all codeword bits $j \in W$ such that $C_j \restriction_{\*x_{\overline{J}}=\rho}$ is a constant function.
\end{lemma}
\begin{proof}
Let $J$ be a random subset formed by selecting each $i \in [n]$ independently with probability $1/3$. For each $j \in \overline{J}$, set $\rho_j = 0$ or $\rho_j = 1$ with probability $1/2$. We have $\E[|J|]=n/3$, and hence the Chernoff bound shows that $|J| < n/6$ with probability $\exp(-\Omega(n))$. Furthermore, for each $j \in W$, $C_j\restriction_{\*x_{\overline{J}}=\rho}$ becomes a constant function except with probability $\delta/8$.  This is because for each $i \in T_j$, $C_j\restriction_{x_i=0}$ or $C_j\restriction_{x_i=1}$ is a constant function, and either case happens with probability $1/3$. Therefore
\begin{align*}
	\Pr\left[ C_j\restriction_{\*x_{\overline{J}}=\rho}\textup{ is not constant} \right] \le \tp{1-\frac{1}{3}}^{|T_j|} < e^{-|T_j|/3} \le \frac{\delta}{8}, 
\end{align*} 
where the last inequality is due to $w_j = |T_j| \ge 3\ln(8/\delta)$, since $j \in W$.

By linearity of expectation and Markov's inequality, we have
\begin{align*}
	  & \Pr\left[ \sum_{j \in W}\mathbf{1}\set{C_j\restriction_{\*x_{\overline{J}}=\rho}\textup{ is not constant}} \ge \frac{\delta}{4}|W| \right] \\
	\le& \frac{ \E\left[ \sum_{j \in W}\mathbf{1}\set{C_j\restriction_{\*x_{\overline{J}}=\rho}\textup{ is not constant}} \right]}{\delta|W|/4} \\
	=& \frac{\sum_{j \in W}\Pr\left[ C_j\restriction_{\*x_{\overline{J}}=\rho}\textup{ is not constant} \right]}{\delta|W|/4} \\
	\le& \frac{\delta/8 \cdot |W|}{\delta|W|/4} \le \frac{1}{2}.
\end{align*}
Applying a union bound gives
\begin{align*}
	\Pr\left[ \tp{|J| < n/6} \lor \tp{\sum_{j \in W}\mathbf{1}\set{C_j\restriction_{\*x_{\overline{J}}=\rho}\textup{ is not constant}} \ge \frac{\delta}{4}|W|} \right] \le \exp\tp{-\Omega(n)} + \frac{1}{2} < 1.
\end{align*}
Finally, we can conclude that there exist $J \subseteq [n]$ and $\rho\in\zo^{\overline{J}}$ such that $|J| \ge n/6$, and $C_j\restriction_{\*x_{\overline{J}}=\rho}$ becomes a constant function for all but $\delta/4$ fraction of $j \in W$. 
\end{proof}

Let $J \subseteq [n]$ and $\rho \in \set{0,1}^{\overline{J}}$ be given by the \cref{lem:random-restriction}, and consider the restricted code $C_{J|\rho}$. By rearranging the codeword bits, we may assume $J=[n']$ where $n'=|J| \ge n/6$. 
Let $A \subseteq [m]$ be the set of codeword bits which get fixed to constants under $J|\rho$. We denote $W'\coloneqq W\setminus A$, $S_i'\coloneqq S_i \setminus A$, $S_{i,-}'\coloneqq S_{i,-}\setminus A$, and $S_{i,+}'\coloneqq S_{i,+}\setminus A$. Note that $|W'|=|W\setminus A| \le \delta m/4$, and thus $|S_{i,+}'|=|S_{i,+}\cap W'|\le \delta m/4$ for all $i \in [n']$. We emphasize that $S_i'$ does not necessarily contain all codeword bits fixable by $x_i$ in the restricted code $C_{J|\rho}$, as fixing some message bits may cause more codeword bits to be fixable by $x_i$.

We first show that the queries of $C$ must have certain structures. The following claim characterizes the queries in $F_i^{\ge 1}$.
\begin{claim} \label{clm:bot-fix}
	Suppose $\set{j,k} \in F_{i}^{\ge 1}$. Then we must have $j,k \in S_i$.
\end{claim}
\begin{proof}
	Let $\set{j, k} \in F_{i}^{\ge 1}$. Suppose for the sake of contradiction that $j \notin S_i$. This implies there are partial assignments $\sigma_{00}, \sigma_{01}, \sigma_{10}, \sigma_{11} \in \set{0,1}^{n-1}$ such that 
	\begin{align*}
	C_j\tp{\*x_{-i} = \sigma_{00}, x_i = 0} = 0, \quad C_j\tp{\*x_{-i} = \sigma_{01}, x_i = 1} = 0, \\
	C_j\tp{\*x_{-i} = \sigma_{10}, x_i = 0} = 1, \quad C_j\tp{\*x_{-i} = \sigma_{11}, x_i = 1} = 1,
	\end{align*}
	where $\*x_{-i}$ is defined as $\tp{x_t \colon t \in [n]\setminus\set{i}}$.
	
	Let $C_{00}, C_{01}, C_{10}, C_{11}$ be encodings of the corresponding assignments mentioned above. Since the relaxed decoder has perfect completeness, when $\Dec(i,\cdot)$ is given access to $C_{00}$ or $C_{10}$  it must output $x_i=0$. Note that the $j$-th bit is different in $C_{00}$ and $C_{10}$. Similarly, when $\Dec(i,\cdot)$ is given access to $C_{01}$ or $C_{11}$ it must output $x_i = 1$. However, this already takes up 4 entries  in the truth table of any decoding function $f \in \mathtt{DF}_{j,k}^i$, leaving no space for any ``$\perp$'' entry. This contradicts with the assumption $\set{j,k} \in F_{i}^{\ge 1}$. 
\end{proof}

Here is another way to view \cref{clm:bot-fix} which will be useful later: Suppose $\set{j, k}$ is a query set such that $j \notin S_i$ (or $k \notin S_i$), then $\set{j, k} \in F_{i}^{0}$. In other words, conditioned on the event that some query is not contained in $S_i$, the decoder never outputs $\perp$.

The following claim characterizes the queries in $F_i^{0}$.
\begin{claim} \label{clm:fixable-or-useless}
Suppose $\set{j,k} \in F_{i}^{0}$, and $j \in S_i$. Then one of the following three cases occur: (1) $k \in S_i$, (2) $C_j=x_i$, or (3) $C_j=\neg x_i$.
\end{claim}
\begin{proof}
Since $j \in S_i$, we may, without loss of generality, assume that $C_j\restriction_{x_i=0}$ is a constant function. Let us further assume it is the constant-zero function. The proofs for the other cases are going to be similar. 

Denote by $f(y_j, y_k)$ the function returned by $\Dec(i,\cdot)$ conditioned on reading $\set{j,k}$. Any function $f \in \mathtt{DF}_{j,k}^i$ takes values in $\set{0,1}$ since $\set{j,k} \in F_i^{0}$. Suppose case (1) does not occur, meaning that $C_k\restriction_{x_i=0}$ is not a constant function. Then there must be partial assignments $\sigma_{00}, \sigma_{01} \in \set{0,1}^{n-1}$ such that
\begin{align*}
    C_k(x_i=0, \*x_{-i}=\sigma_{00}) = 0, \quad C_k(x_i=0, \*x_{-i}=\sigma_{01}) = 1.
\end{align*}
Let $C_{00}$ and $C_{01}$ be the encodings of the corresponding assignments mentioned above. Due to perfect completeness of $\Dec$, it must always output $x_i = 0$ when given access to $C_{00}$ or $C_{01}$. That means $f(0,0)=f(0,1)=0$.

Now we claim that $C_j\restriction_{x_i=1}$ must be the constant-one function. Otherwise there is a partial assignment $\sigma_{10} \in \zo^{n-1}$ such that
\begin{align*}
    C_j(x_i=1, \*x_{-i}=\sigma_{10}) = 0.
\end{align*}
Let $C_{10}$ be the encoding of this assignment. On the one hand, due to perfect completeness $\Dec(i,\cdot)$ should always output $x_i=1$ when given access to $C_{10}$. On the other hand, $\Dec(i,\cdot)$ outputs $f((C_{10})_j,0)=f(0,0)=0$. This contradiction shows that $C_j\restriction_{x_i=1}$ must be the constant-one function. Therefore $C_j=x_i$, i.e., case (2) occurs.

Similarly, when $C_j\restriction_{x_i=0}$ is the constant-one function, we can deduce that $C_j=\neg x_i$, i.e., case (3) occurs.
\end{proof}

We remark that \cref{clm:bot-fix} and \cref{clm:fixable-or-useless} jointly show that for any query set $\set{j,k}$ made by $\Dec(i,\cdot)$ there are 2 essentially different cases: (1) both $j, k$ lie inside $S_i$, and (2) both $j, k$ lie outside $S_i$. The case $j \in S_i, k\notin S_i$ ($k \in S_i, j\notin S_i$, resp.) means that $k$ ($j$, resp.) is a dummy query which is not used for decoding. Furthermore, conditioned on case (2), the decoder never outputs $\perp$.

Another important observation is that all properties of the decoder discussed above hold for the restricted code $C_{J|\rho}$, with $S_i$ replaced by $S_i'$. This is because $C_{J|\rho}$ uses essentially the same decoder, except that it does not actually query any codeword bit which became a constant.

For a subset $S \subseteq [m]$,
we say ``$\Dec(i,\cdot)$ reads $S$'' if the event ``$j \in S$ and $k \in S$'' occurs where $j, k \in [m]$ are the queries made by $\Dec(i,\cdot)$. The following lemma says that conditioned on $\Dec(i,\cdot)$ reads some subset $S$, there is a way of modifying the bits in $S$ that flips the output of the decoder.

\begin{lemma} \label{lem:conditional-zero}
Let $S \subseteq [m]$ be a subset such that $\Pr[\Dec(i,\cdot)\textup{ reads }S]>0$. Then for any string $s \in \set{0,1}^m$ and any bit $b \in \set{0,1}$, there exists a string $z \in \set{0,1}^m$ such that $z[[m]\setminus S]=s[[m]\setminus S]$, and 
\begin{align*}
	\Pr\left[ \Dec(i, z) = 1-b \mid \Dec(i,\cdot)\textup{ reads }S \right] = 1.
\end{align*}
\end{lemma}
\begin{proof}
Let $x \in \set{0,1}^{n}$ be a string with $x_i=1-b$. Let $z \in \set{0,1}^{m}$ be the string satisfying 
\begin{align*}
z[S] = C(x)[S], \quad z[[m]\setminus S] = s[[m]\setminus S]. 
\end{align*}
Since $\Dec$ has perfect completeness, we have
\begin{align*}
	1 = \Pr\left[ \Dec(i, C(x)) = x_i \mid \Dec(i,\cdot)\textup{ reads }S \right] = \Pr\left[ \Dec(i, z) = 1-b \mid \Dec(i,\cdot)\textup{ reads }S \right].
\end{align*}
\end{proof}

The next lemma is a key step in our proof. It roughly says that there is a local decoder for $x_i$ in the standard sense as long as the size of $S_{i}$ is not too large.

\begin{lemma} \label{lem:LDC-reduction}
	Suppose $i \in [n]$ is such that $\abs{S_{i}} \le \delta m/2$. Then there is a $(2,\delta/2,1/2+\eps)$-local decoder $D_i$ for $i$. In other words, for any $x \in \set{0,1}^{n}$ and $y \in \set{0,1}^{m}$ such that $\ham(C(x), y) \le \delta m/2$, we have
	\begin{align*}
	    \Pr\left[ D_i(y) = x_i \right] \ge \frac{1}{2} + \eps,
	\end{align*}
	and $D_i$ makes at most 2 queries into $y$.
\end{lemma}
\begin{proof}
	Let $i \in [n]$ be such that $\abs{S_{i}} \le \delta m/2$.
	The local decoder $D_i$ works as follows. Given $x \in \set{0,1}^{n}$ and $y \in \set{0,1}^{m}$ such that $\ham(C(x), y) \le \delta m/2$, $D_i$ obtains a query set $Q$ according to the query distribution of $\Dec(i,\cdot)$ conditioned on $Q \subseteq [m]\setminus S_i$. Then $D_i$ finishes by outputting the result returned by $\Dec(i,\cdot)$. 
	
	
	
	Denote by $E_i$ the event ``$\Dec(i,\cdot)$ reads $[m]\setminus S_i$'', i.e., both two queries made by $\Dec(i,\cdot)$ lie outside $S_i$. In order for the conditional distribution to be well-defined, we need to argue that $E_i$ occurs with non-zero probability. Suppose this is not the case, meaning that $Q \cap S_i \neq \varnothing$ for all possible query set $Q$. Let $z \in \set{0,1}^m$ be the string obtained by applying Lemma~\ref{lem:conditional-zero} with $S = S_i$, $s=C(x)$ and $b=x_i$. Claim~\ref{clm:bot-fix} and Claim~\ref{clm:fixable-or-useless} jointly show that either $Q \subseteq S_i$, or the decoder's output does not depend on the answers to queries in $Q \setminus S_i$. In any case, the output of $\Dec(i,z)$ depends only on $z[S_i]$. However, by the choice of $z$ we now have a contradiction since 
	\begin{align*}
		\frac{1}{2} + \eps \le \Pr\left[ \Dec(i,z) \in \set{x_i, \perp} \right] = \Pr\left[ \Dec(i,z) \in \set{x_i, \perp} \mid \Dec(i,\cdot)\textup{ reads }S_i \right] = 0,
	\end{align*}
	where the first inequality is due to $\ham(C(x),z)\le |S_i| < \delta m$ and the relaxed decoding property of $\Dec$.
	
	By definition of $D_i$, it makes at most 2 queries into $y$. Its success rate is given by
	\begin{align*}
	    \Pr[D_i(y) = x_i] = \Pr[\Dec(i,y) = x_i \mid E_i].
	\end{align*}
	Therefore it remains to show that
	\begin{align*}
		\Pr\left[ \Dec(i,y) = x_i \mid E_i \right] \ge \frac{1}{2} + \eps.
	\end{align*} 
	
	Let $z$ be the string obtained by applying Lemma~\ref{lem:conditional-zero} with $S=S_i$, $s=y$ and $b=x_i$. From previous discussions we see that conditioned on $\overline{E_i}$ (i.e., the event $E_i$ does not occur), the output of $\Dec(i,z)$ only depends on $z[S_i]$. Therefore
	\begin{align}
		\Pr\left[ \Dec(i, z) \in \set{x_i, \perp} \mid \overline{E_i} \right] = 1-\Pr\left[ \Dec(i, z) = 1-x_i \mid \overline{E_i} \right] = 0.
		\label{eqn:conditional-zero}
	\end{align}
	We also have that $z$ is close to $C(x)$ since
	\begin{align*}
	\ham(z, C(x)) \le \ham(z, y) + \ham(y, C(x)) \le \abs{S_i} + \delta m/2 \le \delta m.
	\end{align*}
	Thus, the relaxed decoding property of $\Dec$ gives
	\begin{align*}
	\Pr\left[ \Dec(i, z) \in \set{x_i, \perp} \right] \ge \frac{1}{2} + \eps.
	\end{align*}
	On the other hand, we also have
	\begin{align*}
	 & \Pr\left[ \Dec(i, z) \in \set{x_i, \perp} \right] \\
	=& \Pr\left[ \Dec(i, z) \in \set{x_i, \perp} \mid \overline{E_i} \right] \cdot \Pr\left[ \overline{E_i} \right] + \Pr\left[ \Dec(i, z) \in \set{x_i, \perp} \mid E_i \right] \cdot \Pr\left[ E_i \right] \\
	=& \Pr\left[ \Dec(i, z) \in \set{x_i, \perp} \mid \overline{E_i} \right] \cdot \Pr\left[ \overline{E_i} \right] + \Pr\left[ \Dec(i, y) \in \set{x_i, \perp} \mid E_i \right] \cdot \Pr\left[ E_i \right] \tag*{($z[[m]\setminus S_i]=y[[m]\setminus S_i]$)} \\
	=& \Pr\left[ \Dec(i, y) \in \set{x_i, \perp} \mid E_i \right] \cdot \Pr\left[ E_i \right]  \tag*{Equation (\ref{eqn:conditional-zero})}\\
	\le& \Pr\left[ \Dec(i, y) \in \set{x_i, \perp} \mid E_i \right].
	\end{align*}
	Note that by Claim~\ref{clm:bot-fix}, conditioned on $E_i$, $\Dec(i,\cdot)$ never outputs ``$\perp$''. We thus have 
	\begin{align*}
	\Pr\left[ \Dec(i, y) = x_i \mid E_i \right] \ge \frac{1}{2} + \eps.
	\end{align*}
\end{proof}

We remark once again that the above lemma holds for the restricted code $C_{J|\rho}$, with $S_i$ replaced by $S_i'$.

Below we prove an exponential lower bound for non-adaptive 2-query Hamming RLDCs.
\begin{proposition} \label{prop:non-adaptive-2qRLDC}
	Let $C \colon \set{0,1}^n \rightarrow \set{0,1}^m$ be a non-adaptive weak $(2,\delta,1/2+\eps)$-RLDC. Then $m = 2^{\Omega_{\delta,\eps}(n)}$.
\end{proposition}
\begin{proof}
    Let $C_{J|\rho}\colon \set{0,1}^{n'}\rightarrow \set{0,1}^{m'}$ be the restricted code where $J|\rho$ is given by \cref{lem:random-restriction}, and $A \subseteq [m]$ be the set of codeword bits which get fixed to constants. We also let $S_i'\coloneqq S_i\setminus A$, $S_{i,-}'=S_{i,-}\setminus A$, $S_{i,+}'=S_{i,+}\setminus A$.
    
	Denote $T_j'\coloneqq \set{i \in [n'] \colon j \in S_i'}$. Since $S_i'\subseteq S_i$ for each $i$, we also have $T_j'\subseteq T_j$ for each $j$. In particular, for each $j \notin W' \subseteq W$, we have $|T_j'| \le |T_j| \le 3\ln(8/\delta)$. Therefore
	\begin{align*}
	\underset{i \in [n']}{\mathbb{E}}\left[|S_{i,-}'|\right] = \frac{1}{n'}\sum_{i=1}^{n'}|S_{i,-}'| = \frac{1}{n'}\sum_{j \in [m']\setminus W'}|T_j'| \le 3\ln(8/\delta)\cdot \frac{m'}{n'}.
	\end{align*}
	Therefore by Markov's inequality, 
	\begin{align*}
	\underset{i \in [n']}{\Pr}\left[|S_{i,-}'| > \delta m'/4 \right] \le \frac{12\ln(8/\delta)}{\delta n'} = O_{\delta}\tp{\frac{1}{n'}}.
	\end{align*}
	In other words, there exists $I \subseteq [n']$ of size $|I| \ge n'-O_{\delta}(1)$ such that $|S_{i,-}'| \le \delta m'/4$ for all $i \in I$. For any such $i \in I$, we have $|S_i'| = |S_{i,-}'| + |S_{i,+}'| \le \delta m'/4 + \delta m'/4 = \delta m'/2$. By \cref{lem:LDC-reduction}, we can view $C_{J|\rho}$ as a $(2,\delta/2,1/2+\eps)$-LDC for message bits in $I$ (for instance, we can arbitrarily fix the message bits outside $I$), where $|I| > n'-O_{\delta}(1) = \Omega(n)$. Finally, the statement of the proposition follows from \cref{thm:two-query-lb}. 
\end{proof}

\subsection{Lower bounds for adaptive 2-Query Hamming RLDCs} \label{subsec:adaptive-2qRLDC}
Now we turn to the actual proof, which still works for possibly adaptive decoders. Let $C$ be a weak $(2,\delta,1/2+\eps)$-RLDC with perfect completeness. We fix a relaxed decoder $\Dec$ for $C$. Without loss of generality, we assume $\Dec$ works as follows: on input $i \in [n]$, $\Dec(i,\cdot)$ picks the first query $j \in [m]$ according to a distribution $\+D_i$. Let $b \in \set{0,1}$ be the answer to this query. Then $\Dec$ picks the second query $k \in [m]$ according to a distribution $\+D_{i;j,b}$, and obtains an answer $b' \in \set{0,1}$. Finally, $\Dec$ outputs a random variable $X_{i;j,b,k,b'}\in \set{0,1,\perp}$. 

We partition the support of $\+D_i$ into the following two sets:
\begin{align*}
	F_i^{0} &\coloneqq \set{j \in \supp(\+D_i) \colon \forall b,b' \in \set{0,1}, k \in \supp(\+D_{i;j,b,k,b'}), \Pr[X_{i;j,b,k,b'} = \perp] = 0}, \\
	F_i^{>0} &\coloneqq \set{j \in \supp(\+D_i) \colon \exists b,b' \in \set{0,1}, k \in \supp(\+D_{i;j,b,k,b'}), \Pr[X_{i;j,b,k,b'}=\perp] > 0}.
\end{align*}

We will still apply the restriction guaranteed by \Cref{lem:random-restriction} to $C$. The sets $S_i$, $T_j$, $W$, $S_{i,-}$, $S_{i,+}$ (are their counterparts for $C_{J|\rho}$) are defined in the exact same way.

The following claim is adapted from \Cref{clm:bot-fix}.
\begin{claim} \label{clm:adaptive-bot-fix}
	$(\supp(\+D_i) \setminus S_i) \subseteq F_i^{0}$. 
\end{claim}
\begin{proof}
	Let $j \in \supp(\+D_i) \setminus S_i$ and we will show $j \in F_i^{0}$. By the definition of $S_i$, $j \notin S_i$ means that there are partial assignments $\sigma_{00}, \sigma_{01}, \sigma_{10}, \sigma_{11} \in \set{0,1}^{n-1}$ such that 
	\begin{align*}
	C_j\tp{\*x_{-i} = \sigma_{00}, x_i = 0} = 0, \quad C_j\tp{\*x_{-i} = \sigma_{01}, x_i = 1} = 0, \\
	C_j\tp{\*x_{-i} = \sigma_{10}, x_i = 0} = 1, \quad C_j\tp{\*x_{-i} = \sigma_{11}, x_i = 1} = 1,
	\end{align*}
	where $\*x_{-i}$ is defined as $\tp{x_t \colon t \in [n]\setminus\set{i}}$.
	
	Let $C_{00}, C_{01}, C_{10}, C_{11}$ be encodings of the corresponding assignments mentioned above. Consider an arbitrary query $k \in \supp(\+D_{i;j,0})$, and let $b_1', b_2'$ be the $k$-th bit of $C_{00}$ and $C_{01}$, respectively. We note that $X_{i;j,0,k,b_1'}$ is the output of $\Dec(i,C_{00})$ conditioned on the queries $j, k$, and $X_{i;j,0,k,b_2'}$ is the output of $\Dec(i,C_{01})$ conditioned on the queries $j, k$. Due to perfect completeness of $\Dec$, we have
	\begin{align*}
	\Pr[X_{i;j,0,k,b_1'}=0] = 1, \quad \Pr[X_{i;j,0,k,b_2'}=1]=1.
	\end{align*}
	Therefore, it must be the case that $b_1' \neq b_2'$, which implies that $\Pr[X_{i;j,0,k,b'}=\perp]=0$ for any $b' \in \set{0,1}$.
	
	An identical argument shows that $\Pr[X_{i;j,1,k,b'}=\perp]=0$ for any $k \in \supp(\+D_{i;j,1})$ and $b' \in \set{0,1}$. Thus we have shown $j \in F_i^{0}$.
\end{proof}

We remark that the above claim also implies $F_i^{>0} \subseteq S_i$, since $\supp(\+D_i)$ is a disjoint union of $F_i^{0}$ and $F_i^{>0}$. In other words, conditioned on the event that the first query $j$ is not contained in $S_i$, the decoder never outputs $\perp$.

The next claim is adapted from \Cref{clm:fixable-or-useless}.
\begin{claim} \label{clm:adaptive-fixable-or-useless}
Let $j \in \supp(\+D_i)\cap S_i$. For any $b \in \set{0,1}$ one of the following three cases occurs:
\begin{enumerate}
\item $\supp(\+D_{i;j,b}) \subseteq S_i$; 
\item For any $k \in \supp(\+D_{i;j,b}) \setminus S_i$, $\Pr[X_{i;j,b,k,0} = b] = \Pr[X_{i;j,b,k,1} = b] = 1$; 
\item For any $k \in \supp(\+D_{i;j,b}) \setminus S_i$, $\Pr[X_{i;j,b,k,0} = 1-b] = \Pr[X_{i;j,b,k,1} = 1-b] = 1$.
\end{enumerate}
\end{claim}
\begin{proof}
Since $j \in S_i$, we may, without loss of generality, assume that $C_j \restriction_{x_i=0}$ is a constant function. Let us further assume $C_j \restriction_{x_i=0} \;\equiv 0$. The proofs for the other cases are going to be similar.

Suppose $\supp(\+D_{i;j,0}) \not\subseteq S_i$, and let $k \in \supp(\+D_{i;j,0}) \setminus S_i$. By the definition of $S_i$, $k \notin S_i$ means that there are partial assignments $\sigma_{00}, \sigma_{01} \in \set{0,1}^{n-1}$ such that
\begin{align*}
C_k(x_i=0, \*x_{-i}=\sigma_{00}) = 0, \quad C_k(x_i=0, \*x_{-i}=\sigma_{01}) = 1.
\end{align*}
Let $C_{00}$ and $C_{01}$ be the encodings of the corresponding assignments mentioned above. We note that $X_{i;j,0,k,0}$ and $X_{i;j,0,k,1}$ are the outputs of $\Dec(i,C_{00})$ and $\Dec(i,C_{01})$, respectively, conditioned on the queries $j$, $k$. Due to perfect completeness of $\Dec$, we must have 
\begin{align*}
	\Pr[X_{i;j,0,k,0} = 0] = \Pr[X_{i;j,0,k,1} = 0] = 1,
\end{align*}
since both $C_{00}$ and $C_{01}$ encode messages with $x_i = 0$.

Now we claim that $C_j\restriction_{x_i=1} \;\equiv 1$ must hold. Otherwise there is a partial assignment $\sigma_{10} \in \zo^{n-1}$ such that
\begin{align*}
C_j(x_i=1, \*x_{-i}=\sigma_{10}) = 0.
\end{align*}
Let $C_{10}$ be the encoding of this assignment, and let $b' \in \set{0,1}$ be the $k$-th bit of $C_{10}$. On the one hand, $X_{i;j,0,k,b'}$ is the output $\Dec(i,C_{10})$ conditioned on the queries $j$, $k$, and we have just established
\begin{align*}
	\Pr[X_{i;j,0,k,b'} = 0] = 1.
\end{align*} 
On the other hand, $\Dec(i,C_{10})$ should output $x_i=1$ with probability 1 due to perfect completeness. This contradiction shows that $C_j\restriction_{x_i=1} \;\equiv 1$. 

Similarly, suppose $\supp(\+D_{i;j,1}) \not\subseteq S_i$ and let $k \in \supp(\+D_{i;j,1}) \setminus S_i$, meaning that there are partial assignments $\sigma_{10}, \sigma_{11} \in \set{0,1}^{n-1}$ such that
\begin{align*}
C_k(x_i=1, \*x_{-i}=\sigma_{10}) = 0, \quad C_k(x_i=1, \*x_{-i}=\sigma_{11}) = 1.
\end{align*}
Let $C_{10}$ and $C_{11}$ be the corresponding encodings, and note that $X_{i;j,1,k,0}$ and $X_{i;j,1,k,1}$ are the outputs of $\Dec(i,C_{10})$ and $\Dec(i,C_{11})$, respectively, conditioned on the queries $j$, $k$. Perfect completeness of $\Dec$ implies
\begin{align*}
\Pr[X_{i;j,1,k,0} = 1] = \Pr[X_{i;j,1,k,1} = 1] = 1,
\end{align*}
since both $C_{10}$ and $C_{11}$ encode messages with $x_i = 1$.

So far we have shown that for any $b \in \set{0,1}$ such that $\supp(\+D_{i;j,b}) \not\subseteq S_i$, it holds that
\begin{align*}
	\forall k \in \supp(\+D_{i;j,b})\setminus S_i, \quad \Pr[X_{i;j,b,k,0}=b] = \Pr[X_{i;j,b,k,1}=b] = 1,
\end{align*}
provided that $C_j \restriction_{x_i = 0} \;\equiv 0$. In case of $C_j \restriction_{x_i = 0} \;\equiv 1$, we can use an identical argument to deduce that for any $b \in \set{0,1}$ such that $\supp(\+D_{i;j,b}) \not\subseteq S_i$, it holds that
\begin{align*}
\forall k \in \supp(\+D_{i;j,b})\setminus S_i, \quad \Pr[X_{i;j,b,k,0}=1-b] = \Pr[X_{i;j,b,k,1}=1-b] = 1.
\end{align*}
\end{proof}

Here is another way to view \cref{clm:adaptive-fixable-or-useless}: conditioned on the event that the first query $j$ is contained in $S_i$, either the second query $k$ is also contained in $S_i$, or the output $X_{i;j,b,k,b'}$ is independent of the answer $b'$ to query $k$. In either case, the decoder's output depends solely on the $S_i$-portion of the received string.

Once again, the conclusions of \cref{clm:adaptive-bot-fix} and \cref{clm:adaptive-fixable-or-useless} hold for $C_{J|\rho}$, with $S_i$ replaced by $S_i'$.

Finally, we are ready to prove \Cref{thm:main-2qRLDC}. We recall the Theorem below. 
\twoqrldcmain*
\begin{proof}
The proof is almost identical to the one for \Cref{prop:non-adaptive-2qRLDC}. First, we can show that there exists $I \subseteq [n']$ of size $|I| \ge n'-O_{\delta}(1) = \Omega(n)$ such that $|S_{i,-}'| \le \delta m/4$ for all $i \in I$, and hence $|S_i'|=|S_{i,-}'|+|S_{i,+}'|\le \delta m/2$. Second, similar to the proof of \cref{lem:LDC-reduction}, for each $i \in I$ we can construct a decoder $D_i$ for $x_i$ as follows. $D_i$ restarts $\Dec(i,\cdot)$ until it makes a first query $j \in [m']\setminus S_i'$. Then $D_i$ finishes simulating $\Dec(i,\cdot)$ and returns its output. With the help of \Cref{clm:adaptive-bot-fix} and \Cref{clm:adaptive-fixable-or-useless}, the same analysis in \Cref{lem:LDC-reduction} shows that $D_i$ never returns $\perp$, and that the probability of returning $x_i$ is at least $1/2+\eps$. Finally, the theorem follows from \cref{thm:two-query-lb}. 
\end{proof}

\section{Lower Bounds for Strong Insdel RLDCs}\label{sec:srldc}

In this section, we prove \Cref{thm:main-sinsdel RLDC}. We remind the readers that a strong $(q,\delta,\alpha,\rho)$-Insdel RLDC satisfies all 3 conditions in \cref{def:strongRLDC}, and here we are mainly interested in the case where $q$ is a constant and $\alpha=1/2+\beta$ for $\beta > 0$. In fact, \cref{thm:main-sinsdel RLDC} would still hold even without perfect completeness (i.e., Condition 1 in \cref{def:strongRLDC}), as our proof does not rely on this condition. A corollary to this observation is that essentially the same lower bound also holds for strong Insdel RLDCs with adaptive decoders. This is because the Katz-Trevison reduction from adaptive to non-adaptive decoders does preserve Condition 2 and 3 in \cref{def:strongRLDC}, with the same $\rho$ and mildly worse $\alpha=1/2+\beta/2^{q-1}$ (see \cref{foot:kt-obs}).

Our proof relies on the following result, which is implicit in \cite{blocki2021exponential}. \cite{blocki2021exponential} shows an exponential lower bound on the length of constant-query Insdel LDCs. The core of their argument is the construction of an error distribution $\+D$, whereby they derive necessary properties of the code to imply the exponential lower bound. As remarked in Section 4.1 of \cite{blocki2021exponential}, $\+D$ is oblivious to the decoding algorithm. That means their result holds even if the code is required to handle an error pattern much more innocuous than adversarial errors. This stronger statement allows us to define the notion of ``locally decodable on average against $\+D$'', which would otherwise not be well-defined if the adversary is adaptive to the decoding strategy. 
\begin{theorem}[\cite{blocki2021exponential}]\label{thm:lb-channel}
Let $\delta \in (0,1)$ be a constant. There exists a channel $\mathfrak{D}$ for $m$-bit strings with the following properties.
\begin{itemize}
\item For every $s \in \set{0,1}^m$, $\Pr_{s'\sim \mathfrak{D}(s)}[\ED\tp{s', s} > \delta\cdot 2m] < \negl(m)$.
\item Suppose $C \colon \set{0,1}^n \rightarrow \set{0,1}^m$ is a code which is locally decodable on average against $\mathfrak{D}$. Formally, there is a randomized algorithm $\Dec$ satisfying
\begin{align*}
	\forall i \in [n], \quad \Pr_{\substack{x \in \set{0,1}^n \\ y \sim \mathfrak{D}(C(x))}}\left[ \Dec(i, y) = x_i \right] \ge \frac{1}{2} + \eps,
\end{align*} 
where the probability is taken over the uniform random choice of $x \in \set{0,1}^n$, the randomness of $\mathfrak{D}$, and the randomness of $\Dec$. Furthermore, $\Dec$ makes at most $q$ non-adaptive queries into $y$ in each invocation. Then for every $q\ge 2$ there is a constant $\kappa = \kappa(q,\delta,\eps)$ such that
\begin{align*}
m = \exp\tp{\kappa \cdot n^{1/(2q-3)}}.
\end{align*}
\end{itemize}
\end{theorem}

As a side note, in the definition of Insdel LDCs in \cite{blocki2021exponential}, the decoder also has the length of the received string $y$ as an input. However, the channel $\mathfrak{D}$ constructed in \cite{blocki2021exponential} ensures that $|y|=m$ except with exponentially small probability, where $y \sim \mathfrak{D}(C(x))$. For this reason, we omit this extra input as it almost gives no information to the decoder.

The proof of Theorem~\ref{thm:main-sinsdel RLDC} consists of two steps, and they are captured by Lemma~\ref{lem:amplify} and Lemma~\ref{lem:reduce-to-ILDC} below. The first step is a straightforward confidence amplification step which boosts the success rate $\alpha$ by running the decoding algorithm multiple times. In the second step, we show that if $\alpha$ is sufficiently close to 1, a strong Insdel RLDC will imply a standard Insdel LDC that is decodable on average against the channel $\mathfrak{D}$ mentioned in Theorem~\ref{thm:lb-channel}, which is sufficient for deriving the exponential lower bound.

\begin{lemma} \label{lem:amplify}
Let $C \colon \set{0,1}^n \rightarrow \set{0,1}^m$ be a non-adaptive strong $(q,\delta,1/2+\beta,\rho)$-Insdel RLDC where $\beta > 0$. Then for any $\eps > 0$, $C$ is also a non-adaptive strong $(q\cdot \ln(1/\eps)/(2\beta^2), \delta, 1-\eps, \rho)$-Insdel RLDC.
\end{lemma}
\begin{proof}
Let $\Dec$ be a relaxed local decoder for $C$. For some integer $T$ to be decided, consider the following alternative local decoder $\Dec_{T}$ for $C$. On input $(y, m, i)$, $\Dec_T$ independently runs $\Dec(y, m, i)$ for $T$ times, and obtains outputs $r_1, r_2,\dots, r_T \in \set{0,1,\bot}$. For $b \in \set{0,1,\bot}$ we denote 
\begin{align*}
S_b \coloneqq \set{t \in [T] \colon r_t = b}. 
\end{align*}
$\Dec_T$ outputs $0$ or $1$ if $|S_0| \ge T/2$ or $|S_1| \ge T/2$, respectively. Otherwise $\Dec_T$ outputs $\bot$. Clearly, $\Dec_T$ is non-adaptive if $\Dec$ is non-adaptive.

Now we prove the three properties of $\Dec_T$. Perfect completeness is easy to see. The relaxed decoding property is violated when $|S_{1-x_i}| \ge T/2$. By the Chernoff bound, this happens with probability at most $e^{-2\beta^2 T}$, since for each $t \in [T]$ we have 
\begin{align*}
\Pr[r_t = 1-x_i] = 1-\Pr[r_t \in \set{x_i, \perp}] \le 1 - \tp{\frac{1}{2} + \beta} = \frac{1}{2}-\beta.
\end{align*}

Let $I_y \subseteq [n]$ be the subset given by the third property of $\Dec$. That is, for each $i \in I_y$ and $t \in [T]$, we have 
\begin{align*}
\Pr[r_t = x_i] \ge \frac{1}{2}+\beta.
\end{align*}
Again, by the Chernoff bound we have $\Pr[|S_{x_i}| < T/2] < e^{-2\beta^2 T}$, for each $i \in I_y$.

Finally, we take $T = \ln(1/\eps)/(2\beta^2)$ which ensures $e^{-2\beta^2 T} \le \eps$. We note that $\Dec_{T}$ makes $q\cdot T = q\cdot \ln(1/\eps)/(2\beta^2)$ queries to $y$.
\end{proof}

\begin{lemma} \label{lem:reduce-to-ILDC}
Let $C \colon \set{0,1}^n \rightarrow \set{0,1}^m$ be a non-adaptive $(q,\delta,\alpha,\rho)$-Insdel RLDC. Suppose $\rho\alpha + (1-\rho)\alpha/2 = 1/2+\eps$ for some $\eps > 0$. Then there exists a non-adaptive decoder $\Dec$ and a subset $I \subseteq [n]$ of size at least $\eps n/2$ such that for every $i \in I$, we have 
\begin{align*}
	\Pr_{\substack{x \in \set{0,1}^n \\ y \sim \mathfrak{D}(C(x))}}\left[ \Dec(i, y) = x_i \right] \ge \frac{1}{2} + \frac{\eps}{4}.
\end{align*}
The probability is taken over the uniform random choice of $x \in \set{0,1}^n$, the randomness of $\mathfrak{D}$, and the randomness of $\Dec$.
\end{lemma}
\begin{proof}
Let $\Dec_0$ be the relaxed decoder for $C$. The local decoder $\Dec$ will simulate $\Dec_0$ and output the result, except when $\Dec_0$ returns ``$\perp$'', $\Dec$ instead returns a uniform random bit. Clearly, $\Dec$ is non-adaptive if $\Dec_0$ is non-adaptive.

We note that $\mathfrak{D}$ introduces at most $\delta \cdot 2m$ insertions and deletions except with probability $\negl(m)$. Therefore by definition of strong Insdel RLDCs (specifically Condition 2 and 3 in \cref{def:strongRLDC}), for a random index $i \in [n]$, we have
\begin{align*}
    \Pr_{\substack{i \in [n] \\ x \in \set{0,1}^n \\ y \sim \mathfrak{D}(C(x))}}\left[ \Dec(i, y) = x_i \right] \ge \rho\alpha + \tp{1 - \rho} \cdot \frac{\alpha}{2} - \negl(m) \ge \frac{1}{2} + \frac{\eps}{2},
\end{align*}
for large enough $m$ (and thus $n$). The first inequality is because conditioned on $i \in I_y$ (which happens with probability $\ge \rho$ for any $y$), $\Dec(i, y)=x_i$ with probability $\alpha$; and conditioned on $i \notin I_y$ (which happens with probability $\le 1-\rho$), the random guess provides an additional success rate of $\alpha/2$. We will show that the following subset of indices has large density:
\begin{align*}
I = \set{i \in [n] \colon \Pr_{\substack{x \in \set{0,1}^n \\ y \sim \mathfrak{D}(C(x))}}\left[ \Dec(i, y) = x_i \right] \ge \frac{1}{2} + \frac{\eps}{4}}. 
\end{align*}
Denote by $p = |I|/n$ the density of $I$. We have
\begin{align*}
    \frac{1}{2} + \frac{\eps}{2} &\le \Pr_{\substack{i \in [n] \\ x \in \set{0,1}^n \\ y \sim \mathfrak{D}(C(x))}}\left[ \Dec(i, y) = x_i \right] \\
    &= \Pr[i \in I] \cdot \Pr_{\substack{x \in \set{0,1}^n \\ y \sim \mathfrak{D}(C(x))}}\left[ \Dec(i, y) = x_i \mid i \in I \right] + \Pr[i \notin I]  \cdot \Pr_{\substack{x \in \set{0,1}^n \\ y \sim \mathfrak{D}(C(x))}}\left[ \Dec(i, y) = x_i \mid i \notin I \right] \\
    &\le p\cdot 1 + (1-p)\cdot \tp{\frac{1}{2} + \frac{\eps}{4}} \\
    &\le \frac{1}{2} + \frac{\eps}{4} + \frac{p}{2}.
\end{align*}
It follows that $p \ge \eps/2$.
\end{proof}

Now we are ready to prove \Cref{thm:main-sinsdel RLDC}. We recall the theorem below.
\strirldcmain*
\begin{proof}
We first prove the bound for non-adaptive decoders. Taking $\eps = \rho/4$ in Lemma~\ref{lem:amplify}, we have that $C$ is also a non-adaptive strong $(q',\delta,\alpha', \rho)$-Insdel RLDC, where $q'=q\cdot \ln(4/\rho)/(2\beta^2)$ and $\alpha'=1-\rho/4$. Note that
\begin{align*}
	\rho\alpha' + \frac{(1-\rho)\alpha'}{2} = \frac{(1+\rho)\alpha'}{2} = \frac{(1+\rho)(1-\rho/4)}{2} \ge \frac{1}{2} + \frac{\rho}{4}.
\end{align*}
Thus we can apply Lemma~\ref{lem:reduce-to-ILDC} with $\eps = \rho/4$ to obtain a subset $I\subseteq [n]$ of size $|I| \ge \rho n/8$, and a decoder $\Dec$ satisfying
\begin{align*}
	\Pr_{\substack{x \in \set{0,1}^n \\ y \sim \mathfrak{D}(C(x))}}\left[ \Dec(i, y) = x_i \right] \ge \frac{1}{2} + \eps' \coloneqq \frac{1}{2}+\frac{\rho}{16} 
\end{align*}
for every $i \in I$. By \cref{thm:lb-channel}, for some constant $c_1=c_1(q,\delta,\beta,\rho)$ we have 
\begin{align*}
m \ge \exp\tp{\kappa(q',\delta,\eps') \cdot |I|^{1/(2q'-3)}} \ge \exp\tp{c_1 \cdot n^{1/(2q')}} = \exp\tp{c_1\cdot n^{\beta^2/(q \cdot \ln(4/\rho))}}
\end{align*}
as desired. 

For the adaptive case, we note that the proof does not rely on perfect completeness of the decoder (i.e., Condition 1 in \cref{def:strongRLDC}). Therefore, we can apply the Katz-Trevisan reduction (see \cref{foot:kt-obs}) to obtain a non-adaptive decoder for $C$ which satisfies Condition 2 and 3 in \cref{def:strongRLDC}, with the same $\rho$ and mildly worse $\alpha=1/2+\beta/2^{q-1}$. The proof argument presented in this section still applies to such a non-adaptive decoder, whereby we can derive the same lower bound except with $\beta$ replaced by $\beta/2^{q-1}$.
\end{proof}

We end this section with a remark on linear/affine weak 2-query insdel RLDCs. In \Cref{sec:2qrldc}, a transformation from RLDCs to standard LDCs was given for Hamming errors. This is done by fixing some message bits to 0 or 1, together with other modifications to the decoding algorithm. Here the code will remain affine if the initial code is linear or affine. One key step in the analysis is using the perfect completeness condition to deduce structural properties about the queries. We note that the same argument would yield the same query structure for weak insdel RLDCs. Altogether, this allows us to use the impossibility result for affine 2-query insdel LDCs \cite{blocki2021exponential} to conclude the following.
\begin{corollary} \label{cor:linear-2qIRLDC}
For any linear or affine weak $(2,\delta,\eps)$ insdel RLDC $C \colon \set{0,1}^n \rightarrow \set{0,1}^m$, we have $n=O_{\delta,\eps}(1)$.
\end{corollary}


\section{Weak Insdel RLDC and Strong Insdel RLCC Constructions}\label{sec:wrLDC}
We prove \cref{thm:main-rildc-construction,thm:main-rilcc-construction} in this section.
As a reminder, a weak $(q, \delta, \alpha)$-Insdel RLDC satisfies the first two conditions of \cref{def:strongRLDC}.
Our constructions will have constant locality $q$, constant error-tolerance $\delta$, and codeword length $m = O(n^{1+\gamma})$ for any $\gamma \in (0,1)$.

\paragraph{Notation.}
We introduce some additional notation we use throughout this section.
We say that a set of integers $I \subset \mathbb{Z}$ of size $n$ is an \emph{interval} if $I = \{a, a+1, \dotsc, a+(n-1)\}$ for some integer $a$.
For two strings $x, y \in \zo^*$, we let $\lcs(x,y) \in \zo^*$ denote the \emph{longest common sub-sequence} between $x$ and $y$. 
We also let $\LCS(x,y)$ denote a matching between $x$ and $y$ given by $\lcs(x,y)$. 
Formally, $\LCS(x,y)$ denotes a sequence of tuples $(i_1, j_1), \dotsc, (i_k, j_k)$, where $k = |\lcs(x,y)|$ satisfying $i_1 < i_2 < \dotsc < i_k \leq |x|$, $j_1 < j_2 < \dotsc j_k \leq |y|$, and $x_{i_\ell} = y_{j_\ell}$ for all $\ell \in [k]$.
Note that $\lcs(x,y)$ need not be unique, but we can always fix one.
It is well-known that $\ED(x,y) = |x| + |y| - 2 \cdot |\lcs(x,y)|$. 


\subsection{Encoding Algorithm}\label{sec:wlrdc-enc}
The main ingredients of our encoding algorithm consist of an outer code $\Cout \colon \set{0,1}^n \rightarrow \set{0,1}^{k}$ and an inner code $\Cin$ which we shall view as a mapping $\Cin \colon [k] \times \set{0,1} \rightarrow \set{0,1}^{t}$. 
The encoding of $x \in \set{0,1}^n$ is given by 
\begin{align*}
	C(x) \coloneqq \Cin(1, y_1) \circ 0^t \circ \Cin(2, y_2) \circ 0^t \circ \dots \circ 0^t \circ \Cin(k, y_k)\; ,
\end{align*}
where each $y_j \in \set{0,1}$ is obtained by writing $\Cout(x)=y_1 \circ y_2 \circ \cdots \circ y_{k}$. 
Assuming the rate of $\Cin$ is a constant $\rin=\lceil 1+\log_2{k} \rceil /t > 0$, the length of this encoding $C(x)$ is $m\coloneqq (2k-1) t = \Theta(k\log(k))$. 
We note that the concatenation introduces buffers of 0s (i.e., the string $0^t$), which slightly deviates from the standard code concatenation for Hamming errors. 
We present our formal encoding algorithm in \cref{alg:encoder}.
Next we instantiate the concatenation framework by picking specific outer and inner codes.

\paragraph{The outer code.} We take the outer code $\Cout \colon \zo^n \rightarrow \zo^{k}$ to be a non-adaptive weak $(\qout,\dout,1/2 + \epsout)$-relaxed Hamming LDC given by \cref{lem:weak-ham-rldc} with $\gammaout < \gamma$ and $\epsout = 2\eps$, where $\gamma$ and $\eps$ are given in \cref{thm:main-rildc-construction}.
In particular, we have that $k = O(n^{1+\gammaout})$ and $k \log(k) = O(n^{1+\gamma})$.


\paragraph{The inner code.} We take the inner code $\Cin \colon [k] \times \zo \rightarrow \zo^t$ to be an insertion-deletion code (\ie, it is a non-local code) due to Schulman and Zuckerman \cite{SchZuc99}, given by \cref{lem:sz-code}.
In particular, $\Cin$ has the following properties:
\begin{enumerate}
	\item $\Cin$ has constant rate $\rin = 1/\beta > 0$, where $\beta$ is given in \cref{lem:sz-code}. 
	
	\item $\Cin$ has constant minimum (normalized) edit distance $\din \in (0,1/2)$. 
	
	\item For any interval $I \subseteq [t]$ with $|I|\ge 2$ and any $(j,y) \in [k]\times\set{0,1}$, it holds that $\textsf{wt}(\Cin(j,y)_{I}) \ge \floor{|I|/2}$, where $\textsf{wt}(\cdot)$ denotes the Hamming weight.
\end{enumerate}

	\begin{algorithm}[tp]
		\caption{Encoding algorithm $C$.}\label{alg:encoder}
		\SetKwInOut{Input}{Input}
		\SetKwInOut{Output}{Output}
		\SetKwInOut{Hardcoded}{Hardcoded}
		\Input{A message $x \in \zo^{n}$.}
		\Output{A codeword $c \in \zo^{m}$.}
		\Hardcoded{An outer encoder $\Cout \colon \zo^n \rightarrow \zo^k$ and an inner encoder $\Cin \colon [k] \times \zo \rightarrow \zo^t$.}
		
		Compute $y \coloneqq y_1 \circ y_2 \circ \cdots \circ y_k = \Cout(x)$.
		
		\ForEach{$j \in [k]$}{
			Compute $c_j \coloneqq \Cin(j, y_j)$.\label{line:inner-code-block}
		}
	
		Define $c \coloneqq c_1 \circ 0^t \circ c_2 \circ 0^t \circ \cdots \circ 0^t \circ c_k$. \label{line:final-codeword}
		
		\Return{c}
	\end{algorithm}

With our choices of outer and inner codes, we prove the following key lemma about the resulting concatenation code $C$ (\ie, \cref{alg:encoder}).
\begin{lemma}\label{lem:self-nonsimilarity}
	For any interval $I \subseteq [m]$ of length at most $(2-\din)t$, and any index $j \in [k]$, we have $\ED\tp{C[I], \Cin(j, 1-y_j)} \ge \din t/2$.
\end{lemma}
\begin{remark}
	Note that for any interval $I$, we can lower bound the edit distance between $C[I]$ and $\Cin(j, 1-y_j)$ by $\abs{|C_I| - |\Cin(j,1-y_j)} = \abs{|I| - t}$. 
	For any $I$ of length greater than $(2-\din) t$, this implies a lower bound of $(1-\din) t \geq \din t / 2$ for all $\din \leq 1/2$.
\end{remark}
\begin{proof}
Consider an arbitrary LCS matching $M \coloneqq \LCS(C[I], \Cin\tp{j, 1-y_j})  \subseteq I \times [t]$ between $C[I]$ and $\Cin\tp{j,1-y_j}$. 
We prove that $|M| \le (1-\din/2)t$, from which the lemma follows. 

We introduce some notation first.
Given our LCS matching $M$ and a interval $J \subseteq I$, we let $M(J) \coloneqq \{ i \in J \colon \exists j \in [t], (i,j) \in M\}$ denote the set of indices in $J$ that exist in the matching $M$. 
Let $N(J) \coloneqq \set{j \in [t] \colon \exists i \in M(J), (i,j) \in M} \subseteq [t]$ the set of indices in $[t]$ which are matched with $M(J)$. 
Finally, we denote by $\textsf{Span}(J)$ the smallest interval that covers the set $N(J)$. 
We extend all of the definitions to unions of intervals as follows:\footnote{$\textsf{Span}(J_1\cup J_2)$ may not be well-defined if $J_1\cup J_2$ is itself an interval. However, in this paper we will only use this notation for disjoint and non-adjacent $J_1$ and $J_2$, in which case there is a unique way to partition $J_1 \cup J_2$ into disjoint intervals.}
\begin{align*}
	M(J_1 \cup J_2) = M(J_1) \cup M(J_2), \ N(J_1 \cup J_2) = N(J_1) \cup N(J_2), \ \textsf{Span}(J_1 \cup J_2) = \textsf{Span}(J_1) \cup \textsf{Span}(J_2).
\end{align*}
See \Cref{fig:wrldc-notation} for a pictorial overview of $M(J)$, $N(J)$, and $\Span(J)$.
Since $M$ is a monotone matching, we have that $J_1 \cap J_2 = \varnothing$ implies $\Span(J_1) \cap \Span(J_2) = \varnothing$.
We also have $|M(J)|=|N(J)|\le|\textsf{Span}(J)|$.

\begin{figure}[tp]
\begin{center}
\begin{adjustbox}{max width={\linewidth}}
\begin{tikzpicture}[>={stealth}, square/.style={regular polygon,regular polygon sides=4}]
	\tikzset{
		lsquare/.style = 
			{font=\ttfamily\LARGE, draw, square, minimum width=3.5em, text=black},
		xlabel/.style = 
			{label=90:{\Large#1}},
		ylabel/.style = 
			{label=-90:{\Large#1}},
		mj/.style = 
			{preaction={fill=gray, fill opacity=0.3}, pattern=north west lines, pattern color=black!80!white},
		nj/.style = 
			{preaction={fill=gray, fill opacity=0.3}, pattern=north east lines, pattern color=black!80!white},
	}
	\node[lsquare, xlabel=$x_{1}$] at (0,0) (x1) {z};
	\node[lsquare, right = of x1, xlabel=$x_{2}$, mj] (x2) {a};
	\node[lsquare, right = of x2, xlabel=$x_{3}$, mj] (x3) {b};
	\node[lsquare, right = of x3, xlabel=$x_{4}$] (x4) {c};
	\node[lsquare, right = of x4, xlabel=$x_{5}$] (x5) {x};
	\node[lsquare, right = of x5, xlabel=$x_{6}$] (x6) {d};
	
	\node[lsquare, below = 1.5 of x1, ylabel=$y_1$, nj] (y1) {a};
	\node[lsquare, right = of y1, ylabel=$y_2$] (y2) {e};
	\node[lsquare, right = of y2, ylabel=$y_3$] (y3) {b};
	\node[lsquare, right = of y3, ylabel=$y_4$, nj] (y4) {b};
	\node[lsquare, right = of y4, ylabel=$y_5$] (y5) {c};
	\node[lsquare, right = of y5, ylabel=$y_6$] (y6) {d};
	
	\draw[thick] (x2.south) -- (y1.north);
	\draw[thick] (x3.south) -- (y4.north);
	\draw[thick] (x4.south) -- (y5.north);
	\draw[thick] (x6.south) -- (y6.north);
	
	\node[left = of x1, align=left] (list) {$M = \{ (2,1), (3,4), (4,5), (6,6) \}$ \\ \\ $I = \{1,2,3,4,5,6\}$, $J = \{2, 3\}$ \\ \\ $M(J) = \{ 2, 3 \}$, $N(J) = \{1,4\}$ \\ \\ $\Span(J) = \{1,2,3,4\}$}; 
	
	\node[below = 0.5 of list.south west, anchor=north west] (MJ) {Element of $M(J)$:};
	\node[right = 0.1 of MJ, draw, square, mj, minimum width=1.7em] {};
	\node[below = 0.1 of MJ] (NJ) {Element of $N(J)$:};
	\node[right = 0.1 of NJ, draw, square, nj, minimum width=1.7em] {};
	
	\node[below = 0.2 of y1] (rechelp) {};
	\node[draw, rectangle, dashed, fit=(y1)(y4)(rechelp), label=-90:{$\Span(J)$}] {};
	
\end{tikzpicture}
\end{adjustbox}
\end{center}
\caption{Pictorial representation of $M(J)$, $N(J)$, and $\Span(J)$. Here, $x = x_1 \circ \cdots x_6$ and $y = y_1 \circ \cdots y_6$, $\lcs(x,y) = \texttt{abcd}$ and $M = \LCS(x,y) = \{ (2,1), (3,4), (4,5), (6,6) \}$. We take $I = \{1,\cdots, 6\}$ and $J = \{2,3\}$ as an example. In the figure, the edges between nodes represent the matching $M$.}\label{fig:wrldc-notation}
\end{figure}
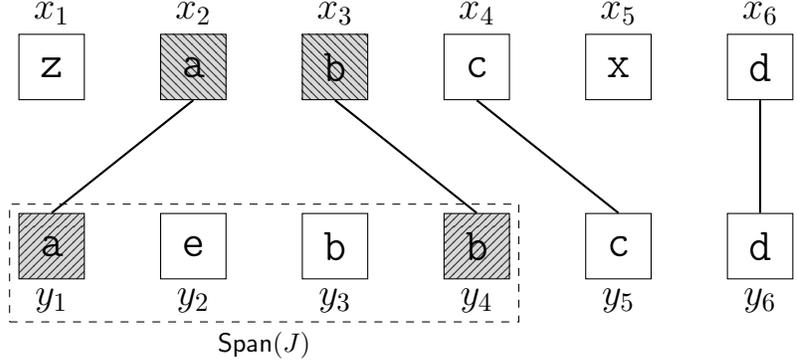

Now we turn back to the proof. 
Since $|I|< 2t$, the interval $I$ spans across at most 2 buffers and/or codewords. 
It is convenient to partition $I$ into $I_b \cup I_c$, where $I_b$ and $I_c$ are unions of at most 2 intervals which correspond to buffers and codewords, respectively.  

We first show that $\abs{\textsf{Span}(I_b)} \ge 2\abs{M(I_b)}$. Intuitively, this is because the density of ``1'' is at least $1/2$ in any interval of an inner codeword, so every matched ``0'' in a buffer has to be accompanied with an insertion of ``1''. Formally, due to Property 3 of $\Cin$, we have
\begin{align*}
	\abs{\textsf{Span}(I_b)} \le 2\textsf{wt}\bigg(\Cin(j, 1-y_j)[\textsf{Span}(I_b)]\bigg) \le 2\abs{\textsf{Span}(I_b) \setminus N(I_b)} = 2\tp{\abs{\textsf{Span}(I_b)} - \abs{M(I_b)}},
\end{align*}
since any index $j \in N(I_b)$ is matched to an index in a buffer, which is necessarily a ``0''. 

We finish the proof in two cases.

\textit{Case 1:} $I_c$ is the union of two intervals. 
In this case, we can deduce that $I$ completely contains a buffer, \ie, $|I_b|=t$, and that $|I_c|\le |I|-|I_b|\le (1-\din)t$. 

Therefore, we have
\begin{align*}
	2|M| \le 2|M(I_c)| + 2|M(I_b)| \le |I_c| + \abs{\textsf{Span}(I_c)} + \abs{\textsf{Span}(I_b)} \le (1-\din)t + t = 2(1-\din/2)t.
\end{align*}
The last inequality is due to $\abs{\textsf{Span}(I_b)} + \abs{\textsf{Span}(I_c)} \le t$, since $\textsf{Span}(I_b), \textsf{Span}(I_c) \subseteq [t]$ are disjoint intervals.

\textit{Case 2:} $I_c$ is an interval.
In this case, note that $M \cap (I_c \times [t])$ corresponds to a common subsequence between $\Cin(j,1-y_j)$ and some other codeword $\Cin(j', y_{j'})$. 
The distance property of $\Cin$ implies $|M(I_c)| \le (1-\din)t$.
Similarly we can upper bound $|M|$ by
\begin{align*}
	2|M| = 2|M(I_c)| + 2|M(I_b)| \le (1-\din)t + \abs{\textsf{Span}(I_c)} + \abs{\textsf{Span}(I_b)} \le (1-\din)t + t = 2(1-\din/2)t.
\end{align*}
To conclude, we have $|M| \le (1-\din/2)t$ in both cases. It follows that $\ED\tp{C[I], \Cin(j,1-y_j)}\ge \din t/2$.
\end{proof}

\subsection{The Decoding Algorithm}\label{sec:wrldc-dec}
Our goal is to construct a relaxed decoder $\Dec$ that, given input an index $i \in [n]$ and oracle access to some binary string $w$ which is $\delta$-close to some codeword $C(x)$ in edit distance, outputs either the bit $x_i$ or $\bot$ with probability at least $1/2 + \eps$. 
We present our formal decoding algorithm in \cref{alg:decoder}.

In our construction, the decoding algorithm invokes the relaxed decoder for the outer code $\Cout$ while providing it with access to a simulated oracle $\widetilde{y} \in \set{0,1}^k$ using the oracle $w$ which is the corrupted codeword.
At a high level, the decoder operates as follows.
\begin{enumerate}
	\item The decoder ensures that the oracle $w$ has length $m$; else it outputs $\bot$.
	
	\item The decoder invokes $\Cout.\Dec(i)$.
	
	\item For each query $j \in [k]$ received from $\Cout.\Dec(i)$
	\begin{enumerate}
		\item The decoder computes codewords $\Cin(j, 0)$ and $\Cin(j,1)$, and additionally computes the first index $i_0 \in [t]$ such that $\Cin(j,0)[i_0] \neq \Cin(j,1)[i_0]$.
		
		\item The decoder samples $i_1,\dotsc, i_d \in [t]$ uniformly and independently at random.
		
		\item The decoder sets a bit $\tildey_j$ as follows. 
		If $w[ 2(j-1)t + i_\ell ] = \Cin(j,0)[i_\ell]$ for all $\ell \in \{0,1,\dotsc, d\}$, then set $\tildey_j = 0$; else if $w[ 2(j-1)t+ i_\ell ] = \Cin(j,1)[i_\ell]$ for all $\ell \in \{0,1,\dotsc, d\}$, then set $\tildey_j = 1$; else if neither case occurs, abort and output $\bot$.
		
		\item Answer $\Cout.\Dec(i)$ with bit $\tildey_j$ and await the next query.
	\end{enumerate}
	
	\item Output symbol $\tildex = \Cout.\Dec(i)$.
\end{enumerate}
\begin{remark}
    We choose to check that the received word has length $m$ to simplify the analysis.
    However, this is not necessary by the following observations.
    First, if $|w| < m$, if the decoder every queries a symbol $j > |w|$, then we assume the oracle returns $\bot$, at which point our decoder can abort and output $\bot$.
    Second, if $|w| > m$, then our decoder will simply ignore any bits beyond $w_m$.
\end{remark}
	\begin{algorithm}[tp]
		\DontPrintSemicolon
		\caption{Decoding algorithm $\Dec$.}\label{alg:decoder}
		\SetKwInOut{Input}{Input}
		\SetKwInOut{Output}{Output}
		\SetKwInOut{Hardcoded}{Hardcoded}
		\SetKwInOut{Oracle}{Oracle}
		\Input{An index $i \in [n$].}
		\Oracle{Bitstring $w \in \zo^{m'}$ for some $m' \in \mathbb{N}$.}
		\Output{A symbol $\tildex \in \set{0,1,\bot}$.}
		\Hardcoded{Parameter $d \in \mathbb{N}$, outer decoder $\Cout.\Dec \colon [k] \rightarrow \zo$, and inner encoder $\Cin \colon [k] \times \zo \rightarrow \zo^t$.}
		
		\lIf(\tcc*[f]{\textcolor{black}{Query Oracle.} Verify length of $w$.}\label{line:len-check}){$w[m+1] \neq \bot$ or $w[m] = \bot$}{\Return{$\tildex = \bot$}}
		
		\ForEach(\tcc*[f]{Handle adaptive decoding.}){oracle query $j \in [k]$ received from $\Cout.\Dec(i)$}{
			Compute $c_{j,0} = \Cin(j, 0)$ and $c_{j,1} = \Cin(j,1)$.
			
			Compute first index $i_0 \in [t]$ such that $c_{j,0}[i_0] \neq c_{j,1}[i_0]$.\label{line:perfect-completeness}
			
			Sample $d$ values $i_1, i_2, \dotsc, i_d \in [t]$ independently and uniformly at random.\label{line:index-sampling}
			
			Compute \tcc*[f]{\textcolor{black}{Query Oracle.} Check consistency with computed inner codewords.}\label{line:tildey}
			\begin{align*}
				\tildey_j \coloneqq \begin{cases}
					0 & \text{if } w[(j-1)\cdot 2t + i_\ell] = c_{j,0}[i_\ell]~ \forall \ell \in \{0,1,\dotsc, d\}\\
					1 & \text{if } w[(j-1)\cdot 2t + i_\ell] = c_{j,1}[i_\ell]~ \forall \ell \in \{0,1,\dotsc, d\}\\
					\bot & \text{otherwise}
				\end{cases}
			\end{align*}
		
			\lIf(\label{line:bot-check}){$\tildey_j = \bot$}{\Return{$\tildex = \bot$}.} 
		
			Answer query $j$ of $\Cout.\Dec(i)$ with value $\tildey_j$ and await the next query.
		}
	
		\Return{$\tildex = \Cout.\Dec(i) \in \set{0,1,\bot}$.}
	\end{algorithm}

\paragraph{Analysis of the Decoder.}
Clearly, the query complexity of \cref{alg:decoder} is $(d+1) \cdot \qout + 2$.\footnote{If $m$ is hard-coded in the decoder and the length $m'$ of the received word is additionally given as input, then $q = (d+1)\cdot \qout$.} 
We take $\delta\coloneqq \din\dout/128$.
Fix a message $x \in \zo^n$ and oracle string $w$ such that $\ED(C(x), w) \leq \delta \cdot 2m$.
Moreover, let $y = \Cout(x) \in \zo^k$.
For $j \in [k]$, we denote by $I_j$ the interval that correspond to $\Cin(j, y_j)$ in $C(x)$. Formally, 
\begin{align*}
	I_j \coloneqq \set{2(j-1)t+1, \dotsc, 2(j-1)t + t}.
\end{align*}



Let $\LCS \coloneqq \LCS(C(x), w)$.
Note that $|\LCS| \geq (1-\delta)m$ since $\ED(C(x), w) \leq \delta \cdot 2m$.
\begin{definition}
	We say $j\in [k]$ is \emph{dangerous} if $\ED\tp{w[I_j], \Cin(j, 1-y_j)} \le \din t/4$.
\end{definition}
We first show that if a block $j$ is not dangerous, then $\widetilde{y}_j = 1-y_j$ happens with small probability.
\begin{proposition}\label{prop:dangerous-bad}
	If $j$ is not dangerous, then $\Pr[\widetilde{y}_j = 1-y_j] \leq (1-\din/8)^d$.
\end{proposition}
\begin{proof}
	It suffices to show that for a uniformly random $i \in [t]$, we have 
	\begin{align*}
		\Pr\left[w[2(j-1)t+i] = \Cin(j,1-y_j)[i]\right] \le 1-\din/8.
	\end{align*}
	Since $j$ is not dangerous, we have $\ED(w[I_j], \Cin(j,1-y_j))>\din t/4$. Therefore
	\begin{align*}
		\Pr[w[{2(j-1)t+i}] \neq \Cin(j,1-y_j)[i]] &= \frac{1}{t} \cdot \ham\tp{w[{I_j}], \Cin(j,1-y_j)}\\ 
		&\ge \frac{1}{2t} \cdot \ED(w[I_j], \Cin(j,1-y_j))
		> \frac{\din}{8}.\qedhere
	\end{align*}
\end{proof}

Now the key step is to upper bound the number of dangerous blocks.

\begin{lemma} \label{lem:dangerous-ub}
	The total number of dangerous blocks is at most $\dout k/2$.
\end{lemma}
\begin{proof}
	Let $\LCS \coloneqq \LCS(C(x), w)$ denote an arbitrary LCS matching between $C(x)$ and $w$. 
	Let $U_C, U_w \subseteq [m]$ be the sets of unmatched bits in $C(x)$ and $w$, respectively; \ie, for every $i \in U_C$ (resp., $j \in U_w$), we have that $(i,\ell) \not\in \LCS$ (resp, $(\ell, j) \not\in \LCS$) for all $\ell \in [m]$.
	For each $j \in [k]$, define the following set of indices which are matched to $I_j$, i.e.,
	\begin{align*}
		N(I_j) \coloneqq \set{i \in [m] \colon \exists i' \in I_j, (i,i') \in \LCS},
	\end{align*}
	and let $\textsf{Span}_j$ be the smallest interval covering $N(I_j)$. Since $\textsf{LCS}$ is a monotone matching, the intervals $\textsf{Span}_1, \dots, \textsf{Span}_k$ are disjoint. 
	We can thus lower bound the edit distance between $C(x)$ and $w$ by
	\begin{align*}
	\ED\tp{C(x), w} = |U_C| + |U_w| \ge \sum_{j=1}^{k}\abs{U_C \cap \textsf{Span}_j} + \sum_{j=1}^{k}\abs{U_w \cap I_j}.
	\end{align*}
	\begin{claim} \label{clm:dangerous-imply-error}
		If $j$ is dangerous, then either $\abs{U_C \cap \textsf{Span}_j} \ge \din t/8$, or $\abs{U_w \cap I_j} \ge \din t/16$.
	\end{claim}
	\begin{proof}[Proof of Claim~\ref{clm:dangerous-imply-error}]
		Assume for the sake of contradiction that  $\abs{U_C \cap \textsf{Span}_j} < \din t/8$ and $|U_w \cap I_j| < \din t/16$. We first note that
		\begin{align*}
		|\textsf{Span}_j| = |U_C \cap \textsf{Span}_j| + |N(j)| < \din t/8 + t.
		\end{align*}
		We also note that $\LCS \cap (\textsf{Span}_j \times I_j)$ corresponds to a common subsequence between $C(x)[\textsf{Span}_j]$ and $w[I_j]$, which has length at least
		\begin{align*}
			\abs{\LCS \cap (\textsf{Span}_j \times I_j)} = \abs{I_j \setminus U_w} > t-\din t/16.
		\end{align*}
		In other words, we have
		\begin{align*}
			\ED\tp{w[{I_j}], C(x)[{\textsf{Span}_j}]} &< |I_j| + |\textsf{Span}_j| - 2(t-\din t/16) \\ 
			&< t + (1+\din/8)t - 2t + \din t/8 = \din t/4.
		\end{align*}
		Since $j$ is dangerous, we also have $\ED\tp{w[I_j], \Cin(j, 1-y_j)}\le \din t/4$. The triangle inequality thus implies
		\begin{align*}
			\ED\tp{C(x)[{\textsf{Span}_j}], \Cin(j, 1-y_j)} &\le \ED\tp{w[{I_j}], C(x)[{\textsf{Span}_j}]} + \ED\tp{w[{I_j}], \Cin(j, 1-y_j)}\\
			&< \din t/4 + \din t/4 = \din t/2.
		\end{align*}
		However, this contradicts \cref{lem:self-nonsimilarity}.
	\end{proof}
	
	Denote by $D$ the set of dangerous blocks. By \cref{clm:dangerous-imply-error} we have
	\begin{align*}
	\delta \cdot 2m \ge \ED\tp{w, C(x)} \ge \sum_{j \in D}\tp{\abs{U_C \cap \textsf{Span}_j} + \abs{U_w \cap I_j}} \ge |D| \cdot \frac{\din t}{16}.
	\end{align*}
	Plugging in $\delta = \din\dout/128$ and $m=(2k-1)t \leq 2kt$, we obtain $|D| \le \dout k/2$.
\end{proof}

Now we are ready to prove \cref{thm:main-rildc-construction}. We recall the theorem below.
\weakirldcmain*
\begin{proof}
	Consider the concatenation code $C$ (with buffers) of \cref{alg:encoder} and a relaxed decoder $\Dec$ defined in \cref{alg:decoder}. 
	We fix an arbitrary message $x \in \set{0,1}^n$ and a string $w$ such that $\ED\tp{C(x), w}\le \delta \cdot 2m$. 
	Here $\delta \coloneqq \din\dout/128$. 
	We also denote $y \coloneqq \Cout(x) \in \set{0,1}^k$.
	For the remainder of the proof, unless otherwise stated, all lines referenced are from our decoder description in \cref{alg:decoder}.
	
	Perfect completeness follows directly via the index $i_0$ computed in \cref{line:perfect-completeness}, the computation of $\tildey_j$ in \cref{line:tildey}, and the perfect completeness of the outer code $\Cout$.
	We now focus on proving relaxed decoding.
	
	Let $D \subseteq [k]$ be a subset containing the indices of all dangerous blocks in $w$. We have that $|D|\le \dout k/2$ by \cref{lem:dangerous-ub}.
	Recall that in \cref{alg:decoder}, $\Dec$ invokes the relaxed decoder $\Dec_{out}$ of $\Cout$ while providing it with oracle access to some string $\widetilde{y} \in \set{0,1}^k$. 
	Denote $\mathbf{Y}=\tp{Y_1,Y_2,\dots,Y_k}$ where each $Y_j \in \set{0,1,\perp}$ is the result that would have been returned by the decoding algorithm on query $j$ in \cref{line:bot-check} of \cref{alg:decoder} (even if $j$ might not be queried). 
	Since the decoder uses independent samples in \cref{line:index-sampling} of \cref{alg:decoder} for each query $j$, $Y_j$'s are independent random variables. 
	Due to \cref{prop:dangerous-bad}, for every $j \in [k]\setminus D$ it holds that
	\begin{align*}
		\Pr\left[ Y_j = 1-y_j \right] \le (1-\din/8)^d < e^{-d\din/8} = \dout/4, 
	\end{align*}
	where in the last equality follows from choosing $d = 8\ln(4/\dout)/\din$.
	
	Denote $d(\mathbf{Y}, y)=\sum_{j \in [k]}\mathbf{1}\set{Y_j=1-y_j}$. 
	Since $Y_j$'s are independent, an application of the Chernoff bound shows that
	\begin{align*}
		\Pr\left[ d(\mathbf{Y},y) > \dout k \right] \le \Pr\left[ \sum_{j \in [k]\setminus D}\mathbf{1}\set{Y_j=1-y_j}\ge \frac{\dout k}{2} \right] \le \exp\tp{-\dout^2(k-|D|)/8} \le \eps
	\end{align*}
	for large enough $n$ (and thus $k$).
	Given $\mathbf{Y} \in \set{0,1,\perp}^k$, define a string $c(\mathbf{Y}) \in \set{0,1}^k$ as follows: 
	\begin{align*}
		c(\mathbf{Y})_j \coloneqq \begin{cases}
			y_j & \textup{if $Y_j \in \set{y_j, \perp}$} \\
			1-y_j & \textup{if $Y_j = 1-y_j$}
		\end{cases}.
	\end{align*}
	According to \cref{line:bot-check} of \cref{alg:decoder}, $\Dec$ aborts with output $\perp$ as long as $Y_j=\perp$ for any query $j$. 
	On the other hand, if $Y_j \neq \perp$ for all queries $j$, then $\Dec$ should output the result returned by $\Dec_{out}$ as if it had oracle access to $c(\mathbf{Y})$. 
	In any case, we have
	\begin{align*}
		\Pr\left[ \Dec(i, w) \in \set{x_i, \perp} \mid \mathbf{Y} \right] \ge \Pr\left[  \Dec_{out}(i, c(\mathbf{Y})) \in \set{x_i, \perp} \right].
	\end{align*}
	
	To conclude, we have
	\begin{align*}
		\Pr\left[ \Dec(i, w) \in \set{x_i, \perp} \right] &= \E_{\mathbf{Y}}\left[ \Pr\left[ \Dec(i, w) \in \set{x_i, \perp} \mid \mathbf{Y} \right] \right] \\
		&\ge \E_{\mathbf{Y}\colon d(\mathbf{Y}, y)\le \dout k}\left[ \Pr\left[ \Dec(i, w) \in \set{x_i, \perp} \mid \mathbf{Y} \right] \right] \cdot \Pr\left[ d(\mathbf{Y}, y) \le \dout k \right] \\
		&\ge \E_{\mathbf{Y}\colon d(\mathbf{Y}, y)\le \dout k}\left[ \Pr\left[  \Dec_{out}(i, c(\mathbf{Y})) \in \set{x_i, \perp} \right] \right] - \eps.
	\end{align*}
	By definition of $d(\mathbf{Y}, y)$, it holds that $\ham(c(\mathbf{Y}), y)\le \dout k$ whenever $d(\mathbf{Y}, y) \leq \dout k$. 
	Since $\Cout$ is a $(\qout, \dout, \epsout)$-relaxed LDC, it holds that 
	\begin{align*}
		\Pr\left[  \Dec_{out}(i, c(\mathbf{Y})) \in \set{x_i, \perp} \right] \ge \frac{1}{2} + \epsout.
	\end{align*}
	By our choice of $\epsout = 2\eps$, we have that
	\begin{align*}
		\Pr\left[ \Dec(i, w) \in \set{x_i, \perp} \right] \ge \frac{1}{2} + 2\eps - \eps = \frac{1}{2} + \eps.
	\end{align*}
	The query complexity is $q(\delta, \eps, \gamma) \coloneqq (d+1)\cdot\qout + 2 = \Theta\tp{\qout \cdot \log(1/\dout)/\din} = \Theta(1)$.
\end{proof}

\subsection{Strong Insdel Relaxed Locally Correctable Codes}\label{sec:rLCC}

In this section, we show how to tweak the above construction to obtain a strong Insdel relaxed Locally Correctable Code (RLCC) with constant locality when the outer Hamming RLDC is replaced with a Hamming RLCC. 
We first give a formal definition of relaxed Locally Correctable Codes.

\begin{definition}[Relaxed Locally Correctable Code]\label{def:rLCC}
    A $(q,\delta,\alpha,\rho)$-Relaxed Locally Correctable Code (RLCC) $C \colon \Sigma^n \rightarrow \Sigma^m$ is a code for which there exists a randomized decoder that makes at most $q$ queries to the received word $y$, and satisfies the following properties:
    \begin{enumerate}
        \item (Perfect completeness) For every $i \in [m]$, if $y = C(x)$ for some message $x$ then the decoder, on input $i$, outputs $y_i$ with probability $1$.
        
        \item (Relaxed decoding) For every $i \in [m]$, if $y$ is such that $\dist(y, C(x)) \leq \delta$ for some unique $C(x)$, then the decoder, on input $i$, outputs $c_i$ or $\bot$ with probability $\geq \alpha$.
        
        \item (Success rate) For every $y$ such that $\dist(y, C(x)) \leq \delta$ for some unique $C(x)$, there is a set $I$ of size $\geq \rho m$ such that for every $i \in I$ the decoder, on input $i$, correctly outputs $c_i$ with probability $\geq \alpha$.
    \end{enumerate}
    We will denote an RLCC that satisfies all 3 conditions by the notion of Strong RLCC, and one that satisfies the first 2 conditions by the notion of Weak RLCC. Furthermore,  if the $q$ queries are made in advance, before seeing entries of the codeword, then the decoder is said to be {\em non-adaptive}; otherwise, it is called {\em adaptive}.
\end{definition}
As with \cref{def:strongRLDC}, the probabilities in the above definition are taken over the randomness  of the decoding algorithm, and $\dist$ is a normalized metric.
When $\dist$ is the normalized Hamming distance, then we say that the code is a Hamming RLCC; similarly, when $\dist$ is the normalized edit distance, we say that the code is a Insdel RLCC.
To obtain our strong Insdel RLCC, we replace the weak Hamming RLDC of \cref{alg:encoder} with the weak Hamming RLCC of \cref{lem:rlcc} due to \cite{AsadiS21}.

We now prove the following corollary.

\strongirlcc*

\begin{proof}
The construction is based on a slight modification of code presented in Section~\ref{sec:wrLDC}.

For the encoding process, the outer code $\Cout \colon \zo^n \rightarrow \zo^{k}$ is replaced by a non-adaptive weak $(\qout,\dout,1/2+\epsout)$-relaxed Hamming LCC. 
The existence of such a code is guaranteed by Lemma~\ref{lem:rlcc}.
We pick $\epsout = 2\eps$ and $\qout(\delta, \eps, \gamma)$ to be a sufficiently large constant such that $k = {n^{1+O(1/q)}}$ and $k\log k = O(n^{1+\gamma}) $. 
The encoding algorithm is the same as \cref{alg:encoder}. 

To ensure the total length of zero buffers is exactly half of the codeword length, we append another $t$ 0's at the end of the codeword output by \cref{alg:encoder}. 
By the analysis in \cref{sec:wlrdc-enc}, we know that for any message $x\in\zo^n$, the codeword $C(x)$ has length $\Theta(k\log k) = O(n^{1+\gamma})$.

The decoding process is similar to \cref{alg:decoder}. Given an input index $i$, we do the following: 

\begin{enumerate}
    \item If the index $i$ belongs to a zero buffer, the decoder aborts and outputs $0$. 
    
	\item The decoder checks if the oracle $w$ has length $m$. If not, it aborts and outputs $\bot$.
	
	\item Let $\tilde{i}$ be the index of the non-zero block that contains $i$-th symbol. The decoder invokes $\Cout.\Dec(\tilde{i})$.
	
	\item For each query $j \in [k]$ received from $\Cout.\Dec(\tilde{i})$
	\begin{enumerate}
		\item The decoder computes codewords $\Cin(j, 0)$ and $\Cin(j,1)$, and additionally computes the first index $i_0 \in [t]$ such that $\Cin(j,0)[i_0] \neq \Cin(j,1)[i_0]$.
		
		\item The decoder samples $i_1,\dotsc, i_d \in [t]$ uniformly and independently at random.
		
		\item The decoder sets a bit $\tildey_j$ as follows. 
		If $w[ 2(j-1)t + i_\ell ] = \Cin(j,0)[i_\ell]$ for all $\ell \in \{0,1,\dotsc, d\}$, then set $\tildey_j = 0$; else if $w[ (j-1)(2t+1)+ i_\ell ] = \Cin(j,1)[i_\ell]$ for all $\ell \in \{0,1,\dotsc, d\}$, then set $\tildey_j = 1$; else if neither case occurs, abort and output $\bot$.
		
		\item Answer $\Cout.\Dec(i)$ with bit $\tildey_j$ and await the next query.
	\end{enumerate}
	
	\item Let $\tilde{y}_{\tilde{i}} = \Cout.\Dec(\tilde{i})$. Compute the $\tilde{i}$-th non-zero block $\Cin(\tilde{i}, \tilde{y}_{\tilde{i}})$. Output the bit in this block that is in the $i$-th position of the whole codeword. 
\end{enumerate}

For the perfect completeness, we note our decoding algorithm will always output $C(x)_i$ if there are no corruption, which is guaranteed by the perfect completeness of the weak Hamming RLCC.

If the input index $i$ is in one of the zero buffers, the decoder always outputs $0$ correctly. The success rate for those indices is $1$. The relaxed decoding property holds for those indices. 

For $i$ not in a zero buffer, let $\tilde{i}$ be the index of the non-zero block that contains the $i$-th symbol. By the relaxed decoding property of the weak Hamming RLCC, given access to $\mathbf{Y}$ (defined in the proof of Theorem~\ref{thm:main-rildc-construction}), the relaxed decoder $\Dec_{out}$ of $\Cout$ will output $\Cout(x)_{\tilde{i}}$ or $\bot$ with probability at least $1/2+\epsout$. If $\Dec_{out}$ outputs $\Cout(x)_{\tilde{i}}$ correctly, our decoder can also output $C(x)_i$ correctly. From the analysis in the proof of Theorem~\ref{thm:main-rildc-construction}, we know the decoder will output $C(x)_i$ or $\bot$ with success probability at least $1/2 + \eps$. Thus, the relaxed decoding property holds.

Finally, we can let $I_y \subseteq [m]$ be the set of all indices that is in a zero buffer. Since the total length of zero buffers is exactly half of the codeword length, $\abs{I_y} = m/2$. We always output 0 correctly for $i\in I_y$. Thus, our code achieves a success rate $\rho \ge 1/2$.
\end{proof}

\section{Acknowledgements} We are indebted to some  anonymous reviewers who helped us improve the presentation of the paper.

\bibliographystyle{alpha}
\bibliography{references}

\newcommand{\etalchar}[1]{$^{#1}$}
\begin{thebibliography}{AGKM22}

\bibitem[AGKM22]{AlrabiahGKM}
Omar Alrabiah, Venkatesan Guruswami, Pravesh Kothari, and Peter Manohar.
\newblock A near-cubic lower bound for 3-query locally decodable codes from
  semirandom csp refutation.
\newblock {\em Electron. Colloquium Comput. Complex.}, 2022.

\bibitem[ALRW17]{AndoniLRW17}
Alexandr Andoni, Thijs Laarhoven, Ilya~P. Razenshteyn, and Erik Waingarten.
\newblock Optimal hashing-based time-space trade-offs for approximate near
  neighbors.
\newblock In {\em {SODA}}, pages 47--66, 2017.

\bibitem[AS21]{AsadiS21}
Vahid~R. Asadi and Igor Shinkar.
\newblock Relaxed locally correctable codes with improved parameters.
\newblock In Nikhil Bansal, Emanuela Merelli, and James Worrell, editors, {\em
  48th International Colloquium on Automata, Languages, and Programming,
  {ICALP} 2021}, volume 198 of {\em LIPIcs}, pages 18:1--18:12. Schloss
  Dagstuhl - Leibniz-Zentrum f{\"{u}}r Informatik, 2021.

\bibitem[BB21]{block2021private}
Alexander~R. Block and Jeremiah Blocki.
\newblock Private and resource-bounded locally decodable codes for insertions
  and deletions.
\newblock In {\em 2021 IEEE International Symposium on Information Theory
  (ISIT)}, pages 1841--1846, 2021.

\bibitem[BBG{\etalchar{+}}20]{BlockBGKZ20}
Alexander~R. Block, Jeremiah Blocki, Elena Grigorescu, Shubhang Kulkarni, and
  Minshen Zhu.
\newblock Locally decodable/correctable codes for insertions and deletions.
\newblock In {\em {FSTTCS}}, volume 182 of {\em LIPIcs}, pages 16:1--16:17,
  2020.

\bibitem[BCG20]{BhattacharyyaCG20}
Arnab Bhattacharyya, L.~Sunil Chandran, and Suprovat Ghoshal.
\newblock Combinatorial lower bounds for 3-query ldcs.
\newblock In {\em {ITCS}}, volume 151 of {\em LIPIcs}, pages 85:1--85:8.
  Schloss Dagstuhl - Leibniz-Zentrum f{\"{u}}r Informatik, 2020.

\bibitem[BCG{\etalchar{+}}22]{blocki2021exponential}
Jeremiah Blocki, Kuan Cheng, Elena Grigorescu, Xin Li, Yu~Zheng, and Minshen
  Zhu.
\newblock Exponential lower bounds for locally decodable and correctable codes
  for insertions and deletions.
\newblock In {\em 2021 IEEE 62nd Annual Symposium on Foundations of Computer
  Science (FOCS)}, pages 739--750, 2022.

\bibitem[BDSS16]{BhattacharyyaDS16}
Arnab Bhattacharyya, Zeev Dvir, Shubhangi Saraf, and Amir Shpilka.
\newblock Tight lower bounds for linear 2-query lccs over finite fields.
\newblock {\em Comb.}, 36(1):1--36, 2016.

\bibitem[BFLS91]{BabaiFLS91}
L{\'{a}}szl{\'{o}} Babai, Lance Fortnow, Leonid~A. Levin, and Mario Szegedy.
\newblock Checking computations in polylogarithmic time.
\newblock In {\em {STOC}}, pages 21--31, 1991.

\bibitem[BG17]{BhattacharyyaG17}
Arnab Bhattacharyya and Sivakanth Gopi.
\newblock Lower bounds for constant query affine-invariant lccs and ltcs.
\newblock {\em {ACM} Trans. Comput. Theory}, 9(2):7:1--7:17, 2017.

\bibitem[BGGZ21]{BlockiGGZ19}
Jeremiah Blocki, Venkata Gandikota, Elena Grigorescu, and Samson Zhou.
\newblock Relaxed locally correctable codes in computationally bounded
  channels.
\newblock {\em IEEE Transactions on Information Theory}, 67(7):4338--4360,
  2021.

\bibitem[BGH{\etalchar{+}}06]{Ben-SassonGHSV06}
Eli Ben{-}Sasson, Oded Goldreich, Prahladh Harsha, Madhu Sudan, and Salil~P.
  Vadhan.
\newblock Robust pcps of proximity, shorter pcps, and applications to coding.
\newblock {\em {SIAM} J. Comput.}, 36(4):889--974, 2006.
\newblock A preliminary version appeared in the Proceedings of the 36th Annual
  {ACM} Symposium on Theory of Computing (STOC).

\bibitem[BGT17]{bhattacharyya2017lower}
Arnab Bhattacharyya, Sivakanth Gopi, and Avishay Tal.
\newblock Lower bounds for 2-query lccs over large alphabet.
\newblock {\em Approximation, Randomization, and Combinatorial Optimization.
  Algorithms and Techniques}, 2017.

\bibitem[BGZ18]{BrakensiekGZ18}
Joshua Brakensiek, Venkatesan Guruswami, and Samuel Zbarsky.
\newblock Efficient low-redundancy codes for correcting multiple deletions.
\newblock {\em {IEEE} Trans. Inf. Theory}, 64(5):3403--3410, 2018.

\bibitem[BK95]{BlumK95}
Manuel Blum and Sampath Kannan.
\newblock Designing programs that check their work.
\newblock {\em J. {ACM}}, 42(1):269--291, 1995.

\bibitem[BKZ20]{BlockiKZ19}
Jeremiah Blocki, Shubhang Kulkarni, and Samson Zhou.
\newblock {On Locally Decodable Codes in Resource Bounded Channels}.
\newblock In Yael~Tauman Kalai, Adam~D. Smith, and Daniel Wichs, editors, {\em
  1st Conference on Information-Theoretic Cryptography (ITC 2020)}, volume 163,
  pages 16:1--16:23, Dagstuhl, Germany, 2020. Schloss Dagstuhl--Leibniz-Zentrum
  f{\"u}r Informatik.

\bibitem[BLR93]{BlumLR93}
Manuel Blum, Michael Luby, and Ronitt Rubinfeld.
\newblock Self-testing/correcting with applications to numerical problems.
\newblock {\em J. Comput. Syst. Sci.}, 47(3):549--595, 1993.

\bibitem[BRdW08]{Ben-AroyaRW08}
Avraham Ben{-}Aroya, Oded Regev, and Ronald de~Wolf.
\newblock A hypercontractive inequality for matrix-valued functions with
  applications to quantum computing and ldcs.
\newblock In {\em {FOCS}}, pages 477--486. {IEEE} Computer Society, 2008.

\bibitem[BSB{\etalchar{+}}21]{Banaletal-nature2021}
James~L. Banal, Tyson~R. Shepherd, Joseph Berleant, Hellen Huang, Miguel Reyes,
  Cheri~M. Ackerman, Paul~C. Blainey, and Mark Bathe.
\newblock Random access dna memory using boolean search in an archival file
  storage system.
\newblock {\em Nature Materials}, 20:1272--1280, 2021.

\bibitem[CGdW13]{ChenGW13}
Victor Chen, Elena Grigorescu, and Ronald de~Wolf.
\newblock Error-correcting data structures.
\newblock {\em {SIAM} J. Comput.}, 42(1):84--111, 2013.

\bibitem[CGHL21]{ChengGHL21}
Kuan Cheng, Venkatesan Guruswami, Bernhard Haeupler, and Xin Li.
\newblock Efficient linear and affine codes for correcting
  insertions/deletions.
\newblock In {\em {SODA}}, pages 1--20. {SIAM}, 2021.

\bibitem[CGS20]{ChiesaGS20}
Alessandro Chiesa, Tom Gur, and Igor Shinkar.
\newblock Relaxed locally correctable codes with nearly-linear block length and
  constant query complexity.
\newblock In Shuchi Chawla, editor, {\em Proceedings of the 2020 {ACM-SIAM}
  Symposium on Discrete Algorithms, {SODA} 2020, Salt Lake City, UT, USA,
  January 5-8, 2020}, pages 1395--1411. {SIAM}, 2020.

\bibitem[CHL{\etalchar{+}}19]{ChengHLSW19}
Kuan Cheng, Bernhard Haeupler, Xin Li, Amirbehshad Shahrasbi, and Ke~Wu.
\newblock Synchronization strings: Highly efficient deterministic constructions
  over small alphabets.
\newblock In Timothy~M. Chan, editor, {\em Proceedings of the Thirtieth Annual
  {ACM-SIAM} Symposium on Discrete Algorithms}, pages 2185--2204. {SIAM}, 2019.

\bibitem[CJLW18]{ChengJLW18}
Kuan Cheng, Zhengzhong Jin, Xin Li, and Ke~Wu.
\newblock Deterministic document exchange protocols, and almost optimal binary
  codes for edit errors.
\newblock In Mikkel Thorup, editor, {\em {FOCS}}, pages 200--211, 2018.

\bibitem[CJLW19]{ChengJ0W19}
Kuan Cheng, Zhengzhong Jin, Xin Li, and Ke~Wu.
\newblock Block edit errors with transpositions: Deterministic document
  exchange protocols and almost optimal binary codes.
\newblock In {\em {ICALP}}, volume 132 of {\em LIPIcs}, pages 37:1--37:15,
  2019.

\bibitem[CKGS98]{ChorKGS98}
Benny Chor, Eyal Kushilevitz, Oded Goldreich, and Madhu Sudan.
\newblock Private information retrieval.
\newblock {\em J. {ACM}}, 45(6):965--981, 1998.

\bibitem[CL21]{ChengL21}
Kuan Cheng and Xin Li.
\newblock Efficient document exchange and error correcting codes with
  asymmetric information.
\newblock In {\em {SODA}}, pages 2424--2443. {SIAM}, 2021.

\bibitem[CLZ20]{ChengLZ20}
Kuan Cheng, Xin Li, and Yu~Zheng.
\newblock Locally decodable codes with randomized encoding.
\newblock {\em CoRR}, abs/2001.03692, 2020.

\bibitem[DGL04]{GopalanLD04}
Yan Ding, Parikshit Gopalan, and Richard Lipton.
\newblock Error correction against computationally bounded adversaries.
\newblock Manuscript, 2004.

\bibitem[DGL21]{dall2021structural}
Marcel Dall'Agnol, Tom Gur, and Oded Lachish.
\newblock A structural theorem for local algorithms with applications to
  coding, testing, and privacy.
\newblock In {\em Proceedings of the 2021 ACM-SIAM Symposium on Discrete
  Algorithms (SODA)}, pages 1651--1665. SIAM, 2021.

\bibitem[DGY11]{DvirGY11}
Zeev Dvir, Parikshit Gopalan, and Sergey Yekhanin.
\newblock Matching vector codes.
\newblock {\em {SIAM} J. Comput.}, 40(4):1154--1178, 2011.

\bibitem[DSW17]{DvirSW17}
Zeev Dvir, Shubhangi Saraf, and Avi Wigderson.
\newblock Superquadratic lower bound for 3-query locally correctable codes over
  the reals.
\newblock {\em Theory Comput.}, 13(1):1--36, 2017.

\bibitem[Efr12]{Efremenko12}
Klim Efremenko.
\newblock 3-query locally decodable codes of subexponential length.
\newblock {\em {SIAM} J. Comput.}, 41(6):1694--1703, 2012.

\bibitem[Gas04]{Gasarch04}
William~I. Gasarch.
\newblock A survey on private information retrieval (column: Computational
  complexity).
\newblock {\em Bulletin of the {EATCS}}, 82:72--107, 2004.

\bibitem[GHS20]{GuruswamiHS20}
Venkatesan Guruswami, Bernhard Haeupler, and Amirbehshad Shahrasbi.
\newblock Optimally resilient codes for list-decoding from insertions and
  deletions.
\newblock In Konstantin Makarychev, Yury Makarychev, Madhur Tulsiani, Gautam
  Kamath, and Julia Chuzhoy, editors, {\em {STOC}}, pages 524--537. {ACM},
  2020.

\bibitem[GKST06]{GoldreichKST06}
Oded Goldreich, Howard~J. Karloff, Leonard~J. Schulman, and Luca Trevisan.
\newblock Lower bounds for linear locally decodable codes and private
  information retrieval.
\newblock {\em Comput. Complex.}, 15(3):263--296, 2006.

\bibitem[GL18]{GuruswamiL18}
Venkatesan Guruswami and Ray Li.
\newblock Coding against deletions in oblivious and online models.
\newblock In Artur Czumaj, editor, {\em Proceedings of the Twenty-Ninth Annual
  {ACM-SIAM} Symposium on Discrete Algorithms}, pages 625--643. {SIAM}, 2018.

\bibitem[GL19a]{gur2019lower}
Tom Gur and Oded Lachish.
\newblock A lower bound for relaxed locally decodable codes.
\newblock {\em arXiv preprint arXiv:1904.08112}, 2019.

\bibitem[GL19b]{GuruswamiL19}
Venkatesan Guruswami and Ray Li.
\newblock Polynomial time decodable codes for the binary deletion channel.
\newblock {\em {IEEE} Trans. Inf. Theory}, 65(4):2171--2178, 2019.

\bibitem[GL21]{GurL21}
Tom Gur and Oded Lachish.
\newblock On the power of relaxed local decoding algorithms.
\newblock {\em {SIAM} J. Comput.}, 50(2):788--813, 2021.

\bibitem[GM12]{GalM12}
Anna G{\'{a}}l and Andrew Mills.
\newblock Three-query locally decodable codes with higher correctness require
  exponential length.
\newblock {\em {ACM} Trans. Comput. Theory}, 3(2):5:1--5:34, 2012.

\bibitem[GRR20]{GurRR20}
Tom Gur, Govind Ramnarayan, and Ron Rothblum.
\newblock Relaxed locally correctable codes.
\newblock {\em Theory Comput.}, 16:1--68, 2020.

\bibitem[GS16]{Guruswami_Smith:2016}
Venkatesan Guruswami and Adam Smith.
\newblock Optimal rate code constructions for computationally simple channels.
\newblock {\em J. ACM}, 63(4):35:1--35:37, September 2016.

\bibitem[GW17]{guruswami2017deletion}
Venkatesan Guruswami and Carol Wang.
\newblock Deletion codes in the high-noise and high-rate regimes.
\newblock {\em IEEE Transactions on Information Theory}, 63(4):1961--1970,
  2017.

\bibitem[Hae19]{Haeupler19}
Bernhard Haeupler.
\newblock Optimal document exchange and new codes for insertions and deletions.
\newblock In David Zuckerman, editor, {\em {FOCS} 2019, Baltimore, Maryland,
  USA, November 9-12, 2019}, pages 334--347, 2019.

\bibitem[HO08]{HemenwayO08}
Brett Hemenway and Rafail Ostrovsky.
\newblock Public-key locally-decodable codes.
\newblock In {\em Advances in Cryptology - {CRYPTO} 2008, 28th Annual
  International Cryptology Conference, Proceedings}, pages 126--143, 2008.

\bibitem[HOSW11]{HemenwayOSW11}
Brett Hemenway, Rafail Ostrovsky, Martin~J. Strauss, and Mary Wootters.
\newblock Public key locally decodable codes with short keys.
\newblock In {\em 14th International Workshop, {APPROX}, and 15th International
  Workshop, {RANDOM}, Proceedings}, pages 605--615, 2011.

\bibitem[HOW15]{HemenwayOW15}
Brett Hemenway, Rafail Ostrovsky, and Mary Wootters.
\newblock Local correctability of expander codes.
\newblock {\em Inf. Comput.}, 243:178--190, 2015.

\bibitem[HRS19]{HaeuplerRS19}
Bernhard Haeupler, Aviad Rubinstein, and Amirbehshad Shahrasbi.
\newblock Near-linear time insertion-deletion codes and
  (1+\emph{{\(\epsilon\)}})-approximating edit distance via indexing.
\newblock In Moses Charikar and Edith Cohen, editors, {\em {STOC}}, pages
  697--708. {ACM}, 2019.

\bibitem[HS17]{HaeuplerS17}
Bernhard Haeupler and Amirbehshad Shahrasbi.
\newblock Synchronization strings: codes for insertions and deletions
  approaching the singleton bound.
\newblock In Hamed Hatami, Pierre McKenzie, and Valerie King, editors, {\em
  {STOC}}, pages 33--46. {ACM}, 2017.

\bibitem[HS18]{HaeuplerS18}
Bernhard Haeupler and Amirbehshad Shahrasbi.
\newblock Synchronization strings: explicit constructions, local decoding, and
  applications.
\newblock In Ilias Diakonikolas, David Kempe, and Monika Henzinger, editors,
  {\em {STOC}}, pages 841--854. {ACM}, 2018.

\bibitem[HS21]{haeupler2021synchronization}
Bernhard Haeupler and Amirbehshad Shahrasbi.
\newblock Synchronization strings and codes for insertions and deletions -- a
  survey, 2021.

\bibitem[HSS18]{HaeuplerSS18}
Bernhard Haeupler, Amirbehshad Shahrasbi, and Madhu Sudan.
\newblock Synchronization strings: List decoding for insertions and deletions.
\newblock In Ioannis Chatzigiannakis, Christos Kaklamanis, D{\'{a}}niel Marx,
  and Donald Sannella, editors, {\em {ICALP}}, volume 107 of {\em LIPIcs},
  pages 76:1--76:14, 2018.

\bibitem[KdW04]{KerenidisW04}
Iordanis Kerenidis and Ronald de~Wolf.
\newblock Exponential lower bound for 2-query locally decodable codes via a
  quantum argument.
\newblock {\em J. Comput. Syst. Sci.}, 69(3):395--420, 2004.

\bibitem[KMRS17]{KoppartyMRS17}
Swastik Kopparty, Or~Meir, Noga Ron{-}Zewi, and Shubhangi Saraf.
\newblock High-rate locally correctable and locally testable codes with
  sub-polynomial query complexity.
\newblock {\em J. {ACM}}, 64(2):11:1--11:42, 2017.

\bibitem[KS16]{KoppartyS16}
Swastik Kopparty and Shubhangi Saraf.
\newblock Guest column: Local testing and decoding of high-rate
  error-correcting codes.
\newblock {\em {SIGACT} News}, 47(3):46--66, 2016.

\bibitem[KT00]{KatzT00}
Jonathan Katz and Luca Trevisan.
\newblock On the efficiency of local decoding procedures for error-correcting
  codes.
\newblock In {\em {STOC}}, pages 80--86, 2000.

\bibitem[Lev66]{Levenshtein_SPD66}
Vladimir~Iosifovich Levenshtein.
\newblock Binary codes capable of correcting deletions, insertions and
  reversals.
\newblock {\em Soviet Physics Doklady}, 10(8):707--710, 1966.
\newblock Doklady Akademii Nauk SSSR, V163 No4 845-848 1965.

\bibitem[LFKN92]{LundFKN92}
Carsten Lund, Lance Fortnow, Howard~J. Karloff, and Noam Nisan.
\newblock Algebraic methods for interactive proof systems.
\newblock {\em J. {ACM}}, 39(4):859--868, 1992.

\bibitem[Lip94]{Lipton94}
Richard~J. Lipton.
\newblock A new approach to information theory.
\newblock In {\em {STACS}}, pages 699--708, 1994.

\bibitem[LTX19]{LiuTX20}
Shu Liu, Ivan Tjuawinata, and Chaoping Xing.
\newblock On list decoding of insertion and deletion errors.
\newblock {\em CoRR}, abs/1906.09705, 2019.

\bibitem[MBT10]{Mercier2010ASO}
Hugues Mercier, Vijay~K. Bhargava, and Vahid Tarokh.
\newblock A survey of error-correcting codes for channels with symbol
  synchronization errors.
\newblock {\em IEEE Communications Surveys and Tutorials}, 12, 2010.

\bibitem[Mit08]{Mitzenmachen-survey}
Michael Mitzenmacher.
\newblock A survey of results for deletion channels and related synchronization
  channels.
\newblock {\em Probability Surveys}, 6:1--3, 07 2008.

\bibitem[MK05]{Kiwi_expectedlength}
Jiri~Matousek Marcos~Kiwi, Martin~Loebl.
\newblock Expected length of the longest common subsequence for large
  alphabets.
\newblock {\em Advances in Mathematics}, 197(2):480--498, 2005.

\bibitem[MPSW05]{MicaliPSW05}
Silvio Micali, Chris Peikert, Madhu Sudan, and David~A. Wilson.
\newblock Optimal error correction against computationally bounded noise.
\newblock In {\em Theory of Cryptography, Second Theory of Cryptography
  Conference, {TCC} 2005, Cambridge, MA, USA, February 10-12, 2005,
  Proceedings}, pages 1--16, 2005.

\bibitem[OPC15]{Ostrovsky-InsdelLDC-Compiler}
Rafail Ostrovsky and Anat Paskin-Cherniavsky.
\newblock Locally decodable codes for edit distance.
\newblock In Anja Lehmann and Stefan Wolf, editors, {\em Information Theoretic
  Security}, pages 236--249, Cham, 2015. Springer International Publishing.

\bibitem[OPS07]{OPS07}
Rafail Ostrovsky, Omkant Pandey, and Amit Sahai.
\newblock Private locally decodable codes.
\newblock In {\em {ICALP}}, pages 387--398, 2007.

\bibitem[Slo02]{Sloane2002OnSC}
N.J.A. Sloane.
\newblock On single-deletion-correcting codes.
\newblock {\em arXiv: Combinatorics}, 2002.

\bibitem[SS16]{ShaltielS16}
Ronen Shaltiel and Jad Silbak.
\newblock Explicit list-decodable codes with optimal rate for computationally
  bounded channels.
\newblock In {\em Approximation, Randomization, and Combinatorial Optimization.
  Algorithms and Techniques, {APPROX/RANDOM}}, pages 45:1--45:38, 2016.

\bibitem[STV99]{SudanTV99}
Madhu Sudan, Luca Trevisan, and Salil~P. Vadhan.
\newblock Pseudorandom generators without the {XOR} lemma (abstract).
\newblock In {\em {CCC}}, page~4, 1999.

\bibitem[SZ99]{SchZuc99}
L.~J. {Schulman} and D.~{Zuckerman}.
\newblock Asymptotically good codes correcting insertions, deletions, and
  transpositions.
\newblock {\em IEEE Transactions on Information Theory}, 45(7):2552--2557,
  1999.

\bibitem[Tre04]{Tre04-survey}
Luca Trevisan.
\newblock Some applications of coding theory in computational complexity.
\newblock {\em CoRR}, cs.CC/0409044, 2004.

\bibitem[WdW05]{WehnerW05}
Stephanie Wehner and Ronald de~Wolf.
\newblock Improved lower bounds for locally decodable codes and private
  information retrieval.
\newblock In {\em {ICALP}}, volume 3580 of {\em Lecture Notes in Computer
  Science}, pages 1424--1436. Springer, 2005.

\bibitem[Woo07]{Woodruff07}
David~P. Woodruff.
\newblock New lower bounds for general locally decodable codes.
\newblock Technical report, Weizmann Institute of Science, Israel, 2007.

\bibitem[Woo12]{Woodruff12}
David~P. Woodruff.
\newblock A quadratic lower bound for three-query linear locally decodable
  codes over any field.
\newblock {\em J. Comput. Sci. Technol.}, 27(4):678--686, 2012.

\bibitem[Yek08]{Yekhanin08}
Sergey Yekhanin.
\newblock Towards 3-query locally decodable codes of subexponential length.
\newblock {\em J. {ACM}}, 55(1):1:1--1:16, 2008.

\bibitem[Yek12]{Yekhanin12}
Sergey Yekhanin.
\newblock Locally decodable codes.
\newblock {\em Foundations and Trends in Theoretical Computer Science},
  6(3):139--255, 2012.

\bibitem[YGM17]{Olgica17}
S.~M. Hossein~Tabatabaei Yazdi, Ryan Gabrys, and Olgica Milenkovic.
\newblock Portable and error-free dna-based data storage.
\newblock {\em Scientific Reports}, 7:2045--2322, 2017.

\end{thebibliography}


\end{document}